\documentclass[journal,10pt]{IEEEtran}
\usepackage{amsmath,amssymb,mathtools}
\usepackage{graphicx}
\usepackage{tikz}
\usetikzlibrary{arrows.meta,positioning,fit,backgrounds,calc}
\usepackage{booktabs}
\usepackage{multirow}
\usepackage{array}
\usepackage{algorithm}
\usepackage{algpseudocode}
\usepackage{url}
\usepackage{xcolor}
\usepackage{balance}
\usepackage{microtype}
\usepackage{siunitx}
\usepackage{enumitem}
\usepackage{cite}
\usepackage{placeins}
\usepackage[hidelinks]{hyperref}

\graphicspath{{Fig/}}
\setlist[itemize]{leftmargin=*,nosep}

\newtheorem{proposition}{Proposition}
\newtheorem{lemma}{Lemma}

\title{From Intents to Algorithms: Verified Algorithm Discovery for Transport Networks}

\author{
\IEEEauthorblockN{Behnam Ojaghi, Ricard Vilalta, and Raül Muñoz} \\
  \IEEEauthorblockA{CTTC, Castelldefels, Barcelona, Spain\\
   \em \{bojaghi,ricard.vilalta,raul.munoz\}@cttc.es} 
 }

\begin{document}
\maketitle

\begin{abstract}
Intent-based networking decouples desired outcomes from device-level
configuration, but most systems still map intents to parameters of an
algorithm selected in advance. Large language models (LLMs) create an
opportunity to automate algorithm design, yet unrestricted generated code
is unsuitable for transport-network control because feasibility,
reproducibility, and robustness must be enforced independently of the
model. We present VERA-TN, a verification-guided framework that compiles a
network intent into a bounded algorithm-design specification. The target
architecture uses an LLM as a semantic variation operator over typed
request-ordering and path-ranking programs; generated logic remains
separated from a trusted allocator that enforces path validity, latency,
capacity, and single-path constraints. We prove feasibility preservation
under explicit assumptions and establish a sufficient bound for the
lexicographic latency tie-break in the exact reference model. The released
proof-of-concept instantiates the same interface with a bounded
ten-parameter numerical candidate and deterministic replay, rather than a
completed live-LLM/AST study. Across 150 certified held-out cases on a 28-node TEFNET24-derived hierarchy, evolutionary search reaches a mean priority-utility ratio of 0.958, compared with 0.952 for equal-budget random search and 0.940 for priority-greedy routing. The gain over random search is small but statistically detectable (Holm-adjusted p = 0.0083). The candidate does not improve congestion relative to MILP-C, and the effect of failure-aware training is inconclusive at the 0.05 level (p = 0.051). Eight discovery runs on the official national topology and replay on 12 unseen metro-regional topologies show no stable intent-specific specialization. These results support the trust-boundary and numerical-evolution claims but do not establish a benefit from LLM generation.
\end{abstract}


\section{Introduction}
\label{sec:introduction}

\IEEEPARstart{I}{ntent-based} networking (IBN) seeks to replace
device-level configuration with declarative descriptions of the outcomes
that a network should achieve. Instead of specifying individual routes,
resource allocations, or controller commands, an operator can express
requirements such as preserving latency for a critical service, avoiding
congested transport links, prioritizing selected service classes, or
minimizing disruptive reconfiguration. An intent-based network system
then translates, resolves, activates, and assures these requirements
through a closed management loop~\cite{ibnsurvey}. This abstraction is
particularly attractive for beyond-5G and 6G infrastructures, where
transport resources must support heterogeneous services whose objectives
and operating conditions evolve over time.

Existing intent-driven resource-management approaches, however,
typically assume that the \emph{algorithm itself has already been
selected}. An intent can modify objective coefficients, SLA targets,
priorities, traffic demands, or constraints, but the underlying
optimization model, heuristic, or learned policy remains structurally
fixed. Recent intent-based network-slicing frameworks, for example,
translate service objectives into mathematical resource-allocation and
conflict-resolution problems and adapt their parameters as network
conditions change~\cite{ibnstnsm,m2oibns}. Such approaches provide
effective intent-aware optimization, but adaptation occurs primarily
\emph{within} a predefined algorithm.

This distinction becomes important when the appropriate decision logic
depends on the intent itself. Consider two feasible transport policies.
A latency-oriented policy may repeatedly choose the shortest path and
concentrate traffic on a scarce cut, whereas a resilience-oriented
policy may deliberately reserve that cut for critical services.
Similarly, a rule that maximizes the number of admitted requests can
behave differently from one that maximizes priority-weighted admission,
and a congestion-minimizing policy may be undesirable when
reconfiguration stability becomes the dominant objective. Changing the
numerical inputs of a fixed heuristic does not necessarily produce the
decision structure required by these different operational goals.

Large language models (LLMs) create a new possibility: instead of
selecting only the \emph{output} of an optimization procedure, they can
participate in designing the procedure itself. Recent work on
LLM-assisted algorithm design has demonstrated that language models can
generate and improve executable heuristics through evaluator-driven
search. FunSearch evolves a designated function inside a predefined
program skeleton~\cite{funsearch}; EoH jointly evolves heuristic ideas
and executable implementations~\cite{eoh}; ReEvo uses reflective
language feedback to guide heuristic evolution~\cite{reevo}; and
LLaMEA evolves complete metaheuristics from evaluation
feedback~\cite{llamea}. More recently, AlphaEvolve has demonstrated the
potential of evaluator-driven code evolution across mathematical,
scientific, and systems problems~\cite{alphaevolve}. These developments
suggest a progression from
\begin{equation}
\begin{aligned}
    \text{human-designed algorithm}
    &\rightarrow\\
    &\text{LLM-assisted algorithm design}.
\end{aligned}
    \label{eq:intro_algorithm_progression}
\end{equation}

Directly transferring this paradigm to transport-network control is,
however, problematic. A candidate algorithm can achieve a high
evaluation score while still containing invalid logic, numerical
instabilities, hidden dependencies, or operations that should never be
available to model-generated code. More importantly, network resource
allocation contains hard operational constraints. A generated procedure
must not reinterpret link capacities, ignore service-latency bounds,
construct invalid paths, or bypass the controller's enforcement policy.
Performance measured on a finite training set also provides no guarantee
that an algorithm will remain effective when the topology, traffic
matrix, telemetry accuracy, or active intent changes.

A second challenge is defining what should actually be generated.
Allowing an LLM to synthesize an unrestricted end-to-end network
controller unnecessarily combines algorithm design, feasibility
enforcement, resource-state management, and actuation into one
probabilistic component. Conversely, restricting generation to a fixed
linear score whose coefficients are optimized provides strong
containment but corresponds primarily to parameter tuning rather than
algorithm discovery. A useful middle ground requires a sufficiently
expressive structural search space while retaining a deterministic
execution boundary.

This paper introduces \emph{VERA-TN}, a
\emph{verification-guided LLM evolution framework for
intent-conditioned transport-network algorithm discovery}. Rather than
asking the LLM to allocate each transport request directly, VERA-TN
compiles an operator intent into a typed algorithm-design specification
and uses the LLM as a semantic variation operator over a bounded
domain-specific language (DSL). Candidate programs define reusable
request-ordering and candidate-path-ranking functions through controlled
combinations of transport features, thresholds, interactions, and
conditional branches.

Crucially, generated logic does not control feasibility or network
state. Candidate paths are constructed by trusted deterministic code,
and every selected path is independently checked against link-capacity,
latency, path-validity, and single-path requirements before resource
state can change. The resulting separation is
\begin{equation}
\begin{split}
\underbrace{\text{intent}}_{\text{desired behavior}}
&\rightarrow
\underbrace{\text{LLM-generated algorithm}}_{\text{ranking logic}}
\rightarrow
\underbrace{\text{trusted allocator}}_{\text{feasibility}} \\
&\rightarrow
\underbrace{\text{controller}}_{\text{actuation}}.
\end{split}
\label{eq:intro_trust_chain}
\end{equation}

VERA-TN is \emph{verification-guided} rather than universally
``verified.'' The methodology separates several distinct levels of
evidence. Static and runtime verification determine whether a candidate
program is admissible to execute. The trusted allocator preserves the
modeled transport constraints independently of the generated ranking
logic. Exact optimization supplies an external quality reference on
tractable instances. Finally, failure-oriented scenarios examine robustness under analytical link failures and capacity degradation. Network-digital-twin (NDT) evaluation remains a target operational-validation stage.

Another important design objective is to remove the LLM from the
service-critical execution path. Once discovery is complete, the
selected typed abstract syntax tree (AST) is serialized, versioned, and
compiled into deterministic local code. Transport requests can then be
processed without invoking the originating LLM. This separation
amortizes the potentially expensive discovery process over repeated
network operation and allows the resulting algorithm to be inspected,
replayed, audited, and rolled back independently of the model provider.

The main contributions of this work are as follows:

\begin{itemize}

    \item \textbf{Intent-conditioned algorithm-discovery formulation:}
    We formulate transport resource allocation at two levels. The first
    is a path-based exact reference problem for priority-aware admission
    and routing. The second treats the \emph{decision procedure itself}
    as the optimization object, explicitly distinguishing a network
    intent from individual transport service requests.

    \item \textbf{Bounded structural algorithm representation:}
    We introduce a typed intent-to-algorithm intermediate
    representation and a transport-specific DSL in which an LLM can
    modify feature selection, bounded interactions, thresholds,
    lexicographic rules, and conditional decision structure rather than
    only numerical coefficients.

    \item \textbf{Trusted execution and feasibility preservation:}
    We isolate generated ranking logic from path construction, resource
    accounting, hard-constraint enforcement, and controller actuation.
    We establish a feasibility-preservation property showing that any
    candidate executed through the trusted allocator satisfies the
    modeled path, latency, capacity, and single-path constraints under
    the stated assumptions.

    \item \textbf{Verification-guided evolution:}
   We define a target typed-AST procedure that combines static and runtime verification, exact-oracle comparison, structural-diversity preservation, and structured failure-witness feedback. The released study evaluates only the bounded numerical instantiation.

    \item \textbf{Reproducible proof-of-concept evaluation:}
   We release deterministic code, seeds, candidate vectors, solver certificates, raw metrics, and figure-generation scripts. The numerical study uses 150 held-out cases on a TEFNET24-derived ~\cite{tefnet24} hierarchy and an external-validity check on the official national topology and 12 unseen metro-regional topologies. Model-connected typed-AST experiments remain future work.

\end{itemize}

These contributions are intended to answer a question that is different
from conventional LLM-assisted network control: \emph{can a generative
model help discover a reusable transport-network algorithm while the
network retains deterministic authority over what constitutes an
admissible action?}

The remainder of this paper is organized as follows.
Section~\ref{sec:related_work} positions VERA-TN with respect to
LLM-assisted algorithm design, LLM-based networking, intent-driven
optimization, and trusted network validation.
Section~\ref{sec:system_model} defines the intent and
transport-network system model and the reference allocation problem.
Section~\ref{sec:representation} introduces the bounded
intent-to-algorithm representation and trusted execution model.
Section~\ref{sec:verified_evolution} presents the verification-guided
LLM evolution procedure.
Section~\ref{sec:prototype} describes the topology, prototype, and
reproducibility framework.
Section~\ref{sec:evaluation} defines and reports the experimental
evaluation.
Section~\ref{sec:discussion} discusses implications and limitations,
and Section~\ref{sec:conclusion} concludes the paper.

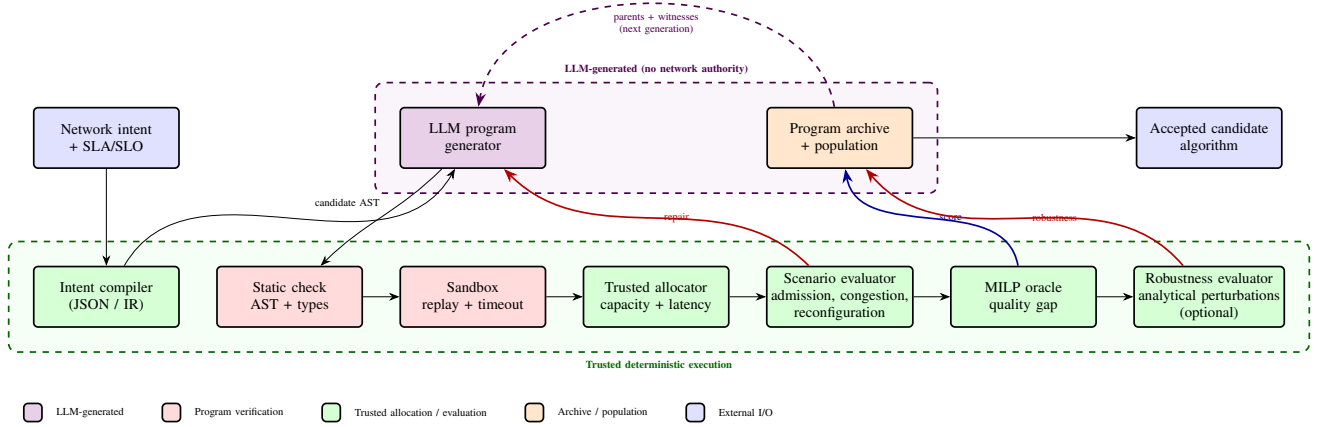
\begin{figure*}[htbp]
    \centering
    \resizebox{0.95\textwidth}{!}{%
    \begin{tikzpicture}[
        >=Stealth,
        every node/.style={font=\scriptsize},
        box/.style={draw, thick, rounded corners=2pt, align=center,
                    minimum height=10mm, minimum width=24mm, inner sep=2pt},
        intentbox/.style={box, fill=blue!12},
        genbox/.style={box, fill=violet!18},
        archbox/.style={box, fill=orange!20},
        verifybox/.style={box, fill=red!14},
        trustbox/.style={box, fill=green!16},
        lbl/.style={font=\tiny},
    ]

    \node[trustbox] (compiler) {Intent compiler\\(JSON / IR)};
    \node[verifybox, right=6mm of compiler] (static) {Static check\\AST + types};
    \node[verifybox, right=6mm of static] (sandbox) {Sandbox\\replay + timeout};
    \node[trustbox, right=6mm of sandbox] (allocator) {Trusted allocator\\capacity + latency};
    \node[trustbox, right=6mm of allocator] (evaluator) {Scenario evaluator\\admission, congestion,\\reconfiguration};
    \node[trustbox, right=6mm of evaluator] (oracle) {MILP oracle\\quality gap};
    \node[trustbox, right=6mm of oracle] (robust) {Robustness evaluator\\analytical perturbations\\(optional)};

    \node[intentbox, above=16mm of compiler] (intent) {Network intent\\+ SLA/SLO};
    \node[genbox, above=16mm of sandbox] (llm) {LLM program\\generator};
    \node[archbox, above=16mm of evaluator] (archive) {Program archive\\+ population};
    \node[intentbox, above=16mm of robust] (accepted) {Accepted candidate\\algorithm};

    \draw[->] (intent) -- (compiler);
    \draw[->] (compiler) to[out=60,in=240] (llm);
    \draw[->] (llm) to[out=225,in=45]
        node[lbl, pos=0.45, above left]{candidate AST} (static);
    \draw[->] (static) -- (sandbox);
    \draw[->] (sandbox) -- (allocator);
    \draw[->] (allocator) -- (evaluator);
    \draw[->] (evaluator) -- (oracle);
    \draw[->] (oracle) -- (robust);
    \draw[->] (archive) -- (accepted);

    \draw[->, thick, red!70!black] (evaluator) to[out=135,in=-45]
        node[lbl, pos=0.5, right, text=red!70!black]{repair} (llm);
    \draw[->, thick, blue!60!black] (oracle) to[out=100,in=-80]
        node[lbl, pos=0.5, right, text=blue!60!black]{score} (archive);
    \draw[->, thick, red!70!black] (robust) to[out=130,in=-50]
        node[lbl, pos=0.5, right, text=red!70!black]{robustness} (archive);

    \draw[->, dashed, thick, violet!60!black] (archive) to[out=100,in=80,looseness=1]
        node[lbl, pos=0.5, below, align=center, text=violet!60!black]
        {parents + witnesses\\(next generation)} (llm);

    \begin{pgfonlayer}{background}
    \node[fit=(compiler)(static)(sandbox)(allocator)(evaluator)(oracle)(robust),
          draw, dashed, thick, green!40!black, rounded corners, fill=green!4,
          inner sep=4mm,
          label={[font=\tiny\bfseries, text=green!40!black]below:Trusted deterministic execution}]
          (trustedbox) {};
    \node[fit=(llm)(archive),
          draw, dashed, thick, violet!60!black, rounded corners, fill=violet!4,
          inner sep=4mm,
          label={[font=\tiny\bfseries, text=violet!60!black]above:LLM-generated (no network authority)}]
          (untrustedbox) {};
    \end{pgfonlayer}

    \node[genbox, minimum width=3mm, minimum height=3mm, inner sep=0]
        (lg1) at ($(compiler.south west)+(0,-14mm)$) {};
    \node[lbl, right=1mm of lg1] (lg1t) {LLM-generated};
    \node[verifybox, minimum width=3mm, minimum height=3mm, inner sep=0, right=5mm of lg1t] (lg2) {};
    \node[lbl, right=1mm of lg2] (lg2t) {Program verification};
    \node[trustbox, minimum width=3mm, minimum height=3mm, inner sep=0, right=5mm of lg2t] (lg3) {};
    \node[lbl, right=1mm of lg3] (lg3t) {Trusted allocation / evaluation};
    \node[archbox, minimum width=3mm, minimum height=3mm, inner sep=0, right=5mm of lg3t] (lg4) {};
    \node[lbl, right=1mm of lg4] (lg4t) {Archive / population};
    \node[intentbox, minimum width=3mm, minimum height=3mm, inner sep=0, right=5mm of lg4t] (lg5) {};
    \node[lbl, right=1mm of lg5] (lg5t) {External I/O};
    \end{tikzpicture}%
    }
    \caption{VERA-TN system architecture and trust boundary. The generated
algorithm (violet) is restricted to request-ordering and path-ranking
logic. Every candidate is verified, executed through the trusted
allocator, scored against the exact oracle, and stress-tested by the
robustness evaluator, all inside the trusted deterministic region
(green). Selected parents and structured failure witnesses are returned
to the generator for the next round of candidate generation, closing the
evolutionary loop. The architectural contribution is therefore
containment of generated decision logic rather than direct LLM control
of network resources.}
    \label{fig:architecture}
\end{figure*}

\section{Related Work}
\label{sec:related_work}

\subsection{Automated Heuristic and Algorithm Design}
Automated heuristic design predates LLM-based generation. Hyper-heuristics search over heuristics rather than directly over problem solutions, while genetic programming has long been used to evolve routing and dispatching rules, including their structure and numerical parameters. VERA-TN does not claim a new evolutionary optimizer. Its contribution is the trusted execution and verification boundary around a generated rule. Evaluator-driven algorithm design has moved from generating isolated code
fragments toward iterative search over executable procedures. FunSearch
evolves a designated function within a trusted program skeleton
\cite{funsearch}; EoH co-evolves heuristic ideas and implementations
\cite{eoh}; ReEvo uses reflective feedback to repair heuristics
\cite{reevo}; and LLaMEA evolves complete metaheuristics under an external
evaluator \cite{llamea}. LLM4AD and recent surveys organize this emerging
design space and its evaluation challenges \cite{llm4ad,surveyAD}, while
AlphaEvolve illustrates the broader potential of evaluator-guided code
evolution \cite{alphaevolve}. VERA-TN adopts the external-evaluator
principle but adds a transport-specific trust boundary: generated logic
may rank requests and prevalidated paths, but it cannot construct paths,
mutate resource state, or invoke a controller
These methods motivate VERA-TN’s typed target interface. However, the present numerical study does not evaluate an LLM and cannot establish whether LLM generation improves on genetic programming under the same grammar and evaluation budget.

\subsection{LLMs and Intent-Based Networking}
LLMs are increasingly used to translate networking tasks, assist protocol
reasoning, and optimize resource allocation. NetLLM develops a general
adaptation framework for networking tasks \cite{netllm}, and LLM-RAO uses
language-model reasoning within resource-allocation optimization
\cite{llmrao}. These methods primarily produce decisions or learned
policies. In parallel, IBN research has established closed-loop intent
translation, conflict resolution, activation, and assurance
\cite{ibnsurvey,intentpaper,ibnconflict}. Recent transport and slicing
frameworks map intents into multi-objective optimization models and
adaptation loops \cite{ibnstnsm,m2oibns}. VERA-TN asks a complementary
question: whether the reusable decision procedure can itself become the
bounded search object while the network retains deterministic enforcement.

\subsection{Traffic Engineering and Trusted Validation}
The allocation model builds on path-based multicommodity flow and
$K$-shortest-path traffic engineering \cite{multicommodity,kshortest}.
Exact mixed-integer optimization supplies a quality reference rather than
an online deployment requirement. The feedback loop is also related to
counterexample-guided inductive synthesis (CEGIS) \cite{cegis}, although
VERA-TN's failure witnesses are empirical counterexamples from finite
network scenarios, not exhaustive symbolic proofs. Network digital twins
provide a complementary environment for perturbation and controller-aware
validation \cite{ndtsurvey}; TeraFlowSDN illustrates an open controller
integration target \cite{teraflow}. The resulting distinction among
program admissibility, allocation feasibility, solution optimality, and
operational validation is central to the proposed methodology.

\section{System Model}
\label{sec:system_model}
The first design question is not which LLM should control the
network, but which decisions a generated algorithm should be allowed
to influence. Fig.~\ref{fig:architecture} therefore separates
\emph{algorithm discovery} from \emph{network authority}. The
generator may modify only the request-ordering and path-ranking logic.
Program admissibility, candidate-path construction, resource accounting,
constraint enforcement, and eventual actuation remain outside the
generated program and are performed by deterministic trusted
components. This boundary is the principal architectural mechanism
through which VERA-TN limits the operational effect of an erroneous
generated candidate.
\subsection{Intent and Transport-Network Model}
\label{subsec:intent_transport_model}

We consider an intent-driven packet transport network represented by an undirected graph
\begin{equation}
    \mathcal{G}=(\mathcal{V},\mathcal{E}),
\end{equation}
where $\mathcal{V}$ denotes the set of transport nodes and $\mathcal{E}$ denotes the set of physical or logical links. Each link $e\in\mathcal{E}$ is characterized by its usable capacity $C_e$, propagation and processing delay $d_e$, current carried load $l_e$, and, when applicable, an operational-risk attribute $q_e$ associated with failures, degradation, or shared-risk conditions. The network state available at a decision epoch is therefore represented by the topology together with the current resource and telemetry state of its links.

The network receives a set $\mathcal{R}$ of transport service requests. Each request $k\in\mathcal{R}$ is represented as
\begin{equation}
    r_k =
    \left(
    s_k,
    t_k,
    b_k,
    L_k^{\max},
    \pi_k,
    c_k
    \right),
    \label{eq:transport_request}
\end{equation}
where $s_k$ and $t_k$ denote the ingress and egress nodes, respectively, $b_k$ is the requested bandwidth, $L_k^{\max}$ is the maximum tolerable end-to-end transport latency, $\pi_k$ represents the service priority, and $c_k$ denotes the service class. The priority parameter does not itself determine a routing decision. Rather, it expresses the relative importance assigned by the operator or service policy when resource contention prevents simultaneous satisfaction of all requests.

Unlike conventional traffic-engineering formulations in which the objective function and algorithm are predefined by the network designer, we assume that network operation is initiated by a declarative intent. An intent specifies the desired network outcome without prescribing the algorithm that must achieve it. We represent a network intent conceptually as
\begin{equation}
    \mathcal{I} =
    \left(
    \mathcal{O},
    \mathcal{C},
    \mathcal{P}
    \right),
    \label{eq:intent_model}
\end{equation}
where $\mathcal{O}$ denotes the ordered set of optimization objectives, $\mathcal{C}$ denotes mandatory operational constraints, and $\mathcal{P}$ contains preference or policy information used to distinguish among feasible solutions.

For example, an operator may express the following intent:
\begin{quote}
\emph{``Maximize the admitted service priority, preserve critical traffic before best-effort traffic, satisfy all bandwidth and latency requirements, and, among solutions with equivalent service admission, minimize peak link utilization.''}
\end{quote}

The distinction between objectives and hard constraints is important. Capacity limits, latency bounds, and path validity are treated as non-negotiable constraints and cannot be relaxed or reinterpreted by an LLM-generated algorithm. Conversely, objectives such as reducing congestion, favoring particular service classes, minimizing reconfiguration, or increasing resilience determine how alternative feasible allocations should be ranked.

Before algorithm discovery, the natural-language intent is compiled into a structured, machine-readable specification. Ambiguous, contradictory, or unsupported requirements are rejected or returned for clarification rather than being resolved implicitly by the generated algorithm. This establishes a stable semantic boundary between what the operator requests and what the algorithm-discovery process is permitted to change.

For each request $k$, the trusted network layer computes a finite set $\mathcal{P}_k$ of up to $K$ loop-free candidate paths between $s_k$ and $t_k$. Candidate paths can be obtained using Yen's $K$-shortest-path algorithm or an equivalent deterministic path-enumeration procedure \cite{kshortest}. The latency associated with candidate path $p\in\mathcal{P}_k$ is
\begin{equation}
    D_{k,p}
    =
    \sum_{e\in p} d_e.
    \label{eq:path_delay}
\end{equation}

Any path for which
\begin{equation}
    D_{k,p} > L_k^{\max}
\end{equation}
is eliminated before algorithm evaluation. Hence, the generated algorithm operates only over paths satisfying the request-level latency constraint.

Let $x_{k,p}\in\{0,1\}$ be a binary decision variable equal to one when request $k$ is assigned to candidate path $p$, and zero otherwise. Let $a_k\in\{0,1\}$ indicate whether request $k$ is admitted. The resulting transport allocation must satisfy three fundamental feasibility conditions: each admitted request is assigned to exactly one path, the aggregate bandwidth placed on every link does not exceed its usable capacity, and only paths satisfying the corresponding latency bound may be selected.


\subsection{Reference Transport Resource-Allocation Problem}
\label{subsec:reference_optimization}

To quantify the quality of discovered algorithms independently of the LLM, we formulate an exact path-based reference optimization problem. Its primary objective is to maximize priority-weighted service admission, with normalized path latency used only as a secondary preference:
\begin{align}
    \max_{\{x_{k,p}\},\{a_k\}}
    \quad &
    \sum_{k\in\mathcal{R}} \pi_k a_k
    -
    \epsilon
    \sum_{k\in\mathcal{R}}
    \sum_{p\in\mathcal{P}_k}
    \frac{D_{k,p}}{L_k^{\max}}x_{k,p},
    \label{eq:reference_objective}
\end{align}
where $\epsilon>0$ controls only the latency tie-break. We use
integer-valued policy priorities, as is common for service classes.

\begin{lemma}[Lexicographic Dominance]
\label{lem:lexicographic}
For a nonempty request set, if every candidate path is latency-feasible and
$0<\epsilon<1/|\mathcal{R}|$, then any one-unit increase in
$\sum_k\pi_k a_k$ dominates every possible change in the normalized
latency term of \eqref{eq:reference_objective}.
\end{lemma}

\begin{IEEEproof}
For an admitted request, exactly one selected path contributes a term in
$[0,1]$ because $D_{k,p}\leq L_k^{\max}$. Hence the total normalized
latency term lies in $[0,|\mathcal{R}|]$, and its change between two
feasible allocations has magnitude at most $|\mathcal{R}|$. Therefore
$\epsilon|\Delta D|<1$, whereas a change in the integer-valued priority
sum is at least one. The priority objective consequently remains
lexicographically dominant.
\end{IEEEproof}
The artifact uses $\epsilon=0.5/(|\mathcal{R}|+1)$, which satisfies this
sufficient bound.

Each admitted request must be assigned to exactly one candidate path:
\begin{equation}
    \sum_{p\in\mathcal{P}_k} x_{k,p}
    =
    a_k,
    \qquad
    \forall k\in\mathcal{R}.
    \label{eq:single_path_assignment}
\end{equation}

The aggregate bandwidth allocated over each transport link cannot exceed its usable capacity:
\begin{equation}
    l_e+
    \sum_{k\in\mathcal{R}}
    \sum_{\substack{p\in\mathcal{P}_k:\\e\in p}}
    b_k x_{k,p}
    \leq
    C_e,
    \qquad
    \forall e\in\mathcal{E}.
    \label{eq:link_capacity}
\end{equation}

The decision variables satisfy
\begin{equation}
    x_{k,p}\in\{0,1\},
    \qquad
    a_k\in\{0,1\}.
    \label{eq:binary_variables}
\end{equation}

Because infeasible-latency paths are removed during candidate-path construction, the formulation does not require an additional latency constraint. Similarly, path continuity and loop freedom are inherited from the trusted path-enumeration procedure.

The reference formulation is not the algorithm deployed by the proposed framework. Instead, it serves as an optimization oracle on instances for which an exact solution can be obtained within the prescribed computational budget. Let $J^\star$ denote its optimal objective value and $J(A)$ denote the objective achieved by candidate algorithm $A$. For oracle-solvable instances, the relative optimality loss of $A$ is measured as
\begin{equation}
    \mathrm{Gap}(A)
    =
    \frac{
    J^\star - J(A)
    }{
    \max\left(\left|J^\star\right|,\epsilon\right)
    }.
    \label{eq:optimality_gap}
\end{equation}

This metric provides an external measure of solution quality that is independent of the LLM's own evaluation or reasoning. For larger instances where proving optimality becomes computationally impractical, the exact oracle can be replaced by the best available solver incumbent together with its optimization bound. Such cases should be reported separately from exact-optimality comparisons.

\subsection{Intent-Conditioned Algorithm-Discovery Problem}
\label{subsec:algorithm_discovery_problem}

The objective of this work is fundamentally different from solving the reference optimization problem once. Given a network-intent specification $\mathcal{I}$, a family of transport-network states $\mathcal{S}$, and a bounded space $\mathcal{A}$ of executable candidate algorithms, we seek an algorithm $A^\star$ whose decision logic performs well across the relevant operating scenarios while remaining inside the trusted feasibility envelope.

The algorithm-discovery problem is expressed as
\begin{align}
    A^\star
    =
    \arg\max_{A\in\mathcal{A}}
    \quad &
    F\left(A\mid\mathcal{I},\mathcal{S}\right)
    \label{eq:algorithm_discovery}
    \\
    \text{s.t.}
    \quad &
    \mathcal{V}
    \left(
    A,\mathcal{I},\mathcal{S}
    \right)
    =
    \mathrm{valid},
    \label{eq:algorithm_validity}
\end{align}
where $F(\cdot)$ evaluates the algorithm according to the objectives encoded by the network intent, and $\mathcal{V}(\cdot)$ represents the structural and operational verification mechanisms described in the following sections.

The output is therefore not a single set of routing decisions. Instead, it is a reusable intent-conditioned algorithm that can subsequently be executed on previously unseen network states without invoking the LLM for every individual service request. This distinction separates the proposed problem from direct LLM-based resource allocation: the LLM participates in discovering and improving the decision procedure, whereas trusted deterministic components execute and verify the resulting procedure during network operation.

The next section defines how the abstract intent in
\eqref{eq:intent_model} is compiled into a bounded
machine-readable representation and specifies the trusted execution
boundary within which candidate algorithms are permitted to operate.

\section{Intent-to-Algorithm Representation and Trusted Execution Model}
\label{sec:representation}

The system model in Section~\ref{sec:system_model} defines the
transport-resource allocation problem and the desired behavior specified
by a network intent. We now define the representation exposed to the
algorithm-discovery process and, critically, the boundary between
LLM-generated decision logic and trusted deterministic network
functions. The design follows the principle that the generated algorithm
may determine \emph{how feasible alternatives are prioritized}, but may
not redefine feasibility, modify network state directly, or bypass
operator constraints.

\subsection{Intent-to-Algorithm Intermediate Representation}
\label{subsec:intent_ir}

Directly exposing an unrestricted natural-language intent to an
algorithm-generation model would leave objectives, constraints, and
available network information open to interpretation. Therefore, before
algorithm discovery, the conceptual intent
$\mathcal{I}=(\mathcal{O},\mathcal{C},\mathcal{P})$ introduced in
\eqref{eq:intent_model} is compiled into the typed intermediate
representation (IR)
\begin{equation}
    \mathcal{I}_{\mathrm{IR}}
    =
    \left\langle
    \mathcal{O},
    \mathcal{C},
    \mathcal{F},
    \mathcal{B},
    \mathcal{D}
    \right\rangle ,
    \label{eq:intent_ir}
\end{equation}
where $\mathcal{O}$ contains the ordered optimization objectives,
$\mathcal{C}$ contains immutable operational constraints,
$\mathcal{F}$ defines the features accessible to generated algorithms,
$\mathcal{B}$ specifies syntactic and computational bounds on candidate
programs, and $\mathcal{D}$ defines the trusted execution and deployment
contract.

The objective component $\mathcal{O}$ specifies \emph{how feasible
allocations are evaluated}, rather than prescribing the algorithm used
to obtain them. For example, the intent considered in
Section~\ref{subsec:intent_transport_model} can be compiled into the
lexicographic hierarchy
\begin{equation}
    \mathcal{O}
    =
    \left[
    \max U_{\mathrm{adm}},
    \min U_{\max},
    \min C_{\mathrm{reconf}}
    \right],
    \label{eq:objective_hierarchy}
\end{equation}
where $U_{\mathrm{adm}}$ denotes priority-weighted service admission,
$U_{\max}$ denotes peak link utilization, and
$C_{\mathrm{reconf}}$ denotes reconfiguration cost when reconfiguration
is considered. Additional objectives can be activated by the intent
without changing the trusted feasibility conditions.

The constraint component $\mathcal{C}$ contains requirements that the
algorithm-discovery process cannot modify. For the transport-network
problem considered here,
\begin{equation}
\begin{aligned}
    \mathcal{C}=\{&\text{capacity},\ \text{latency},\\
    &\text{path validity},\ \text{single-path assignment}\}.
\end{aligned}
    \label{eq:hard_constraint_set}
\end{equation}
If an incoming intent contradicts an operator-defined hard constraint,
the compiler rejects the corresponding specification rather than
implicitly relaxing the constraint during algorithm generation.

The feature schema $\mathcal{F}$ exposes only normalized, read-only
network and request information. Candidate algorithms therefore cannot
query arbitrary controller state or introduce undeclared information
during execution. We partition the feature schema into request-level
and path-level features:
\begin{equation}
    \mathcal{F}
    =
    \mathcal{F}^{\mathrm{req}}
    \cup
    \mathcal{F}^{\mathrm{path}}.
    \label{eq:feature_schema}
\end{equation}

For request $k$, the request-level feature vector is represented as
\begin{equation}
    \mathbf{f}^{\mathrm{req}}_k
    =
    \left[
    \bar{\pi}_k,
    \bar{b}_k,
    \tau_k,
    \chi_k
    \right],
    \label{eq:request_features}
\end{equation}
where $\bar{\pi}_k$ is normalized service priority,
$\bar{b}_k$ is normalized requested bandwidth,
$\tau_k$ is latency tightness, and $\chi_k$ is a numerical or
categorical encoding of the service class.

Latency tightness is defined as
\begin{equation}
    \tau_k
    =
    \frac{
        \min_{p\in\mathcal{P}_k}D_{k,p}
    }{
        L_k^{\max}
    },
    \label{eq:latency_tightness}
\end{equation}
where values approaching one indicate that the request has little
latency margin even on its shortest feasible path.

The bounds component $\mathcal{B}$ restricts generated programs to a
finite, deterministic, and resource-bounded expression space. The
deployment component $\mathcal{D}$ specifies that candidate algorithms
may return only request-ranking and path-ranking scores. They cannot
directly modify link resources, issue controller commands, access the
file system or external network, or call undeclared tools.

Hence, the IR separates three responsibilities:
\begin{equation}
\underbrace{\mathcal{I}_{\mathrm{IR}}}_{\text{desired behavior}}
\longrightarrow
\underbrace{A}_{\text{decision logic}}
\longrightarrow
\underbrace{\mathcal{T}}_{\text{trusted execution}},
\label{eq:intent_algorithm_execution}
\end{equation}
where $\mathcal{T}$ denotes the deterministic transport-resource
allocation layer defined in Section~\ref{subsec:trusted_allocator}.

\subsection{Bounded Algorithm Design Space}
\label{subsec:bounded_design_space}

For a compiled intent $\mathcal{I}_{\mathrm{IR}}$, the discovery process
searches a bounded set
$\mathcal{A}(\mathcal{I}_{\mathrm{IR}})$ of executable candidate
algorithms. Each candidate is represented by two deterministic decision
functions:
\begin{equation}
    A
    =
    \left(
    h_{\mathrm{ord}},
    h_{\mathrm{path}}
    \right),
    \label{eq:candidate_algorithm}
\end{equation}
where $h_{\mathrm{ord}}(k)$ determines the order in which competing
transport requests are considered and
$h_{\mathrm{path}}(k,p)$ ranks candidate path $p$ for request $k$.

Accordingly,
\begin{align}
    h_{\mathrm{ord}}:
    \mathcal{F}^{\mathrm{req}}
    &\rightarrow
    \mathbb{R},
    \label{eq:request_function}
    \\
    h_{\mathrm{path}}:
    \mathcal{F}^{\mathrm{req}}
    \times
    \mathcal{F}^{\mathrm{path}}
    &\rightarrow
    \mathbb{R}.
    \label{eq:path_function}
\end{align}

A central requirement is that the search space allow \emph{structural}
variation rather than only coefficient tuning. We therefore use a typed
domain-specific expression language in which candidate functions may
contain bounded compositions of
\begin{equation}
\begin{split}
    \mathcal{G}_{\mathrm{DSL}}
    =
    \{&
    \texttt{weighted-sum},
    \texttt{lexicographic},
    \texttt{min},
    \texttt{max},\\
    &
    \texttt{bounded-product},
    \texttt{threshold},
    \texttt{conditional}
    \}.
\end{split}
\label{eq:dsl_operators}
\end{equation}

For example, a discovered path-ranking function may instantiate the
conditional structure
\begin{equation}
h_{\mathrm{path}}(k,p)
=
\begin{cases}
g_{\mathrm{tight}}(k,p),
&
\tau_k \geq \delta,
\\[2mm]
g_{\mathrm{flex}}(k,p),
&
\tau_k < \delta,
\end{cases}
\label{eq:conditional_path_rule}
\end{equation}
where $\delta\in[0,1]$ is an evolved threshold,
$g_{\mathrm{tight}}(\cdot)$ may emphasize path latency and bottleneck
utilization, and $g_{\mathrm{flex}}(\cdot)$ may instead emphasize
residual capacity, fragmentation, or operational risk.

Thus, evolution may modify not only numerical coefficients but also
feature selection, comparison order, interaction structure, thresholds,
and conditional branches. This distinction is important because a fixed
functional form with optimized weights would correspond primarily to
parameter tuning, whereas the bounded expression space permits
controlled discovery of distinct decision strategies.

The expression space deliberately excludes unrestricted programming
constructs. Candidate programs cannot contain recursion, unbounded
loops, dynamic imports, mutable global state, arbitrary system calls,
external network access, or controller invocations. Each generated
expression has a statically defined input type and must return a finite
scalar value within a bounded execution budget.

The resulting candidate can be serialized as an expression tree and
replayed independently of the LLM that generated it. This allows
structural comparison among algorithms and supports deterministic
evaluation, archival, and reproducibility.

\subsection{Transport Features for Algorithm Discovery}
\label{subsec:transport_features}

The path-ranking function operates on transport-specific quantities
computed by trusted deterministic components from the current topology
and telemetry state.

For request $k$, admitting it over path $p$ would produce the projected
utilization of link $e\in p$:
\begin{equation}
    u_e(k,p)
    =
    \frac{l_e+b_k}{C_e}.
    \label{eq:projected_link_utilization}
\end{equation}

The projected bottleneck utilization of path $p$ is
\begin{equation}
    u^{\max}_{k,p}
    =
    \max_{e\in p}
    u_e(k,p).
    \label{eq:projected_max_utilization}
\end{equation}

The normalized path-latency consumption is
\begin{equation}
    \rho_{k,p}
    =
    \frac{D_{k,p}}{L_k^{\max}},
    \label{eq:normalized_path_delay}
\end{equation}
where smaller values indicate greater latency headroom.

The normalized residual-capacity slack after admitting request $k$ is
defined as
\begin{equation}
    s_{k,p}
    =
    \min_{e\in p}
    \left(
    \frac{C_e-l_e-b_k}{C_e}
    \right).
    \label{eq:residual_slack}
\end{equation}

To characterize how evenly residual capacity is distributed across a
candidate path, let
\begin{equation}
    r_e(k,p)
    =
    \frac{C_e-l_e-b_k}{C_e}
\end{equation}
denote normalized residual capacity on link $e$ after the prospective
allocation. We define path fragmentation as
\begin{equation}
    \phi_{k,p}
    =
    \frac{1}{|p|}
    \sum_{e\in p}
    \left(
    r_e(k,p)-\bar{r}_{k,p}
    \right)^2,
    \label{eq:fragmentation}
\end{equation}
where
\begin{equation}
    \bar{r}_{k,p}
    =
    \frac{1}{|p|}
    \sum_{e\in p}
    r_e(k,p).
\end{equation}
A larger value of $\phi_{k,p}$ indicates a more uneven residual-resource
profile and therefore a potentially more fragmented path.

When operational-risk information is available, the path-risk feature
is defined as
\begin{equation}
    q_p
    =
    \frac{1}{|p|}
    \sum_{e\in p}q_e.
    \label{eq:path_risk}
\end{equation}

The normalized hop count is
\begin{equation}
    \eta_p
    =
    \frac{|p|}
    {H^{\max}},
    \label{eq:normalized_hops}
\end{equation}
where $H^{\max}$ is the maximum candidate-path length considered in the
current topology.

The resulting path-level feature vector is
\begin{equation}
\mathbf{f}^{\mathrm{path}}_{k,p}
=
\left[
u^{\max}_{k,p},
\rho_{k,p},
s_{k,p},
\phi_{k,p},
q_p,
\eta_p,
c^{\mathrm{reconf}}_{k,p}
\right],
\label{eq:path_feature_vector}
\end{equation}
where $c^{\mathrm{reconf}}_{k,p}$ is included only when the corresponding
intent considers reconfiguration cost.

All quantities in
\eqref{eq:request_features}--\eqref{eq:path_feature_vector}
are read-only from the perspective of the generated algorithm. The
candidate can use them to construct its ranking policy but cannot
directly alter the topology, capacity values, telemetry, or feasibility
criteria.

\subsection{Trusted Allocation Semantics}
\label{subsec:trusted_allocator}

The generated algorithm determines prioritization, whereas resource
allocation itself remains inside a trusted deterministic component.

Let
\begin{equation}
    \sigma_A
    =
    \operatorname{sort}_{\downarrow}
    \left(
    \mathcal{R};
    h_{\mathrm{ord}}
    \right)
    \label{eq:request_order}
\end{equation}
denote the request sequence produced by candidate algorithm $A$, with larger request scores processed first. For
each request $k$ in this sequence, its prevalidated candidate paths are
ordered according to
\begin{equation}
    \Gamma_k^A
    =
    \operatorname{sort}_{\uparrow}
    \left(
    \mathcal{P}_k;
    h_{\mathrm{path}}(k,\cdot)
    \right).
    \label{eq:path_order}
\end{equation}
Smaller path scores are preferred; candidate identifiers provide a deterministic final tie-break.

The trusted allocator evaluates the paths in $\Gamma_k^A$ sequentially.
A candidate path $p$ can be committed only if
\begin{equation}
    l_e+b_k
    \leq
    C_e,
    \qquad
    \forall e\in p.
    \label{eq:trusted_capacity_check}
\end{equation}

Since paths violating
\begin{equation}
    D_{k,p}
    \leq
    L_k^{\max}
    \label{eq:trusted_latency_check}
\end{equation}
have already been removed during candidate-path construction, a path
passing \eqref{eq:trusted_capacity_check} satisfies both the link-capacity
and latency requirements.

When such a path is found, the trusted state transition is
\begin{equation}
    l_e
    \leftarrow
    l_e+b_k,
    \qquad
    \forall e\in p,
    \label{eq:atomic_state_update}
\end{equation}
and the request is marked as admitted. If no path in $\Gamma_k^A$
satisfies the feasibility conditions, request $k$ is rejected.

Therefore, $h_{\mathrm{ord}}$ determines which request receives earlier
access to scarce resources and $h_{\mathrm{path}}$ determines which
feasible route is preferred, but neither function can directly commit an
infeasible allocation.

\subsection{Feasibility-Preservation Property}
\label{subsec:feasibility_property}

The trusted execution boundary provides a feasibility property that is
independent of the quality of the generated algorithm.

\begin{proposition}[Feasibility Preservation]
\label{prop:feasibility}
Assume that:
\begin{enumerate}
    \item the initial network state satisfies
    $l_e\leq C_e$ for every $e\in\mathcal{E}$;
    \item every candidate path $p\in\mathcal{P}_k$ is loop-free,
    connects $s_k$ to $t_k$, and satisfies
    $D_{k,p}\leq L_k^{\max}$;
    \item the trusted allocator admits request $k$ on path $p$ only if
    \eqref{eq:trusted_capacity_check} is satisfied; and
    \item resource-state updates in \eqref{eq:atomic_state_update} are
    atomic.
\end{enumerate}
Then every allocation returned by the trusted allocator satisfies path
validity, request-level latency constraints, single-path assignment, and
link-capacity constraints, independently of the structure or parameters
of $h_{\mathrm{ord}}$ and $h_{\mathrm{path}}$.
\end{proposition}

\textbf{Proof.} By Assumption~2, every path exposed to a generated algorithm is already
a valid, loop-free path connecting the corresponding source and
destination and satisfying its latency bound. Hence, path validity and
latency feasibility cannot be violated by path ranking.

Consider the sequence $\sigma_A$ of requests generated by
$h_{\mathrm{ord}}$. The initial state is feasible by Assumption~1.
Before admitting any request $k$ on path $p$, the trusted allocator
checks \eqref{eq:trusted_capacity_check}. Therefore, if all link loads
satisfy their capacity limits before the admission, they continue to
satisfy them after the atomic update in
\eqref{eq:atomic_state_update}. By induction over all admitted requests,
link-capacity feasibility is preserved throughout the allocation
sequence.

Finally, the allocator commits at most one candidate path for each
request, while rejected requests receive no path. Hence, the
single-path-assignment condition is also satisfied.

Proposition~\ref{prop:feasibility} establishes
\emph{feasibility preservation}, not optimality or complete correctness
of the generated algorithm. A candidate may still make poor admission
choices, create unnecessarily high utilization, produce excessive
reconfiguration, disadvantage lower-priority services, or generalize
poorly to unseen network states. These properties cannot be guaranteed
by the trusted allocator alone. They are therefore evaluated through the
verification-guided evolutionary process introduced in
Section~\ref{sec:verified_evolution}.

\section{Verification-Guided LLM Evolution}
\label{sec:verified_evolution}

Section~\ref{sec:representation} defines the bounded algorithm space
within which an LLM is permitted to modify transport-resource allocation
logic. We now describe how candidate algorithms are generated, verified,
evaluated, and iteratively improved. The key design principle is that
the LLM does not directly optimize a single network instance or issue
network-control actions. Instead, it acts as a semantic variation
operator over a bounded algorithm representation, while deterministic
verification and evaluation components determine whether the resulting
candidate is admissible and useful.

Fig.~\ref{fig:vera_loop} summarizes the discovery loop. Starting from
the compiled intent $\mathcal{I}_{\mathrm{IR}}$, a set of seed
algorithms, and a collection of discovery scenarios, the LLM proposes
new request-ordering and path-ranking programs. Each candidate first
passes structural and runtime verification. Valid candidates are then
executed through the trusted allocator defined in
Section~\ref{subsec:trusted_allocator}. Their behavior is evaluated
against the intent objectives, and, for tractable instances, compared
with the exact optimization oracle in
Section~\ref{subsec:reference_optimization}. Instead of returning only
a scalar fitness value, the evaluator extracts structured failure
witnesses identifying where and why a candidate underperforms. These
witnesses are supplied to subsequent LLM calls to guide mutation,
recombination, or repair.

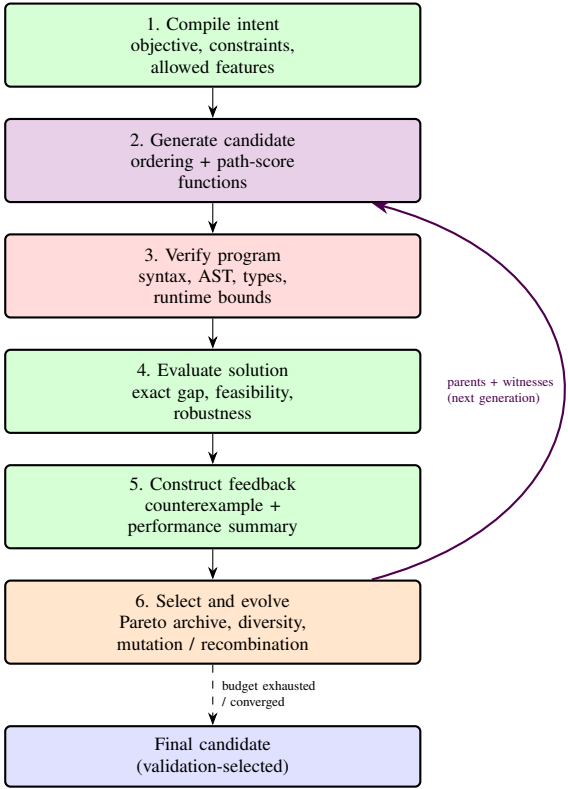
\begin{figure}[htbp]
    \centering
    \begin{tikzpicture}[
        >=Stealth,
        node distance=4mm,
        every node/.style={font=\scriptsize},
        box/.style={draw, thick, rounded corners=2pt, align=center,
                    minimum width=0.62\linewidth, minimum height=11mm, inner sep=2pt},
        lbl/.style={font=\tiny},
    ]
    \node[box, fill=green!16] (compile)
        {1. Compile intent\\objective, constraints,\\allowed features};
    \node[box, fill=violet!18, below=of compile] (generate)
        {2. Generate candidate\\ordering + path-score\\functions};
    \node[box, fill=red!14, below=of generate] (verify)
        {3. Verify program\\syntax, AST, types,\\runtime bounds};
    \node[box, fill=green!16, below=of verify] (evaluate)
        {4. Evaluate solution\\exact gap, feasibility,\\robustness};
    \node[box, fill=green!16, below=of evaluate] (feedback)
        {5. Construct feedback\\counterexample +\\performance summary};
    \node[box, fill=orange!20, below=of feedback] (select)
        {6. Select and evolve\\Pareto archive, diversity,\\mutation / recombination};
    \node[box, fill=blue!12, minimum height=8mm, below=8mm of select] (final)
        {Final candidate\\(validation-selected)};

    \draw[->] (compile) -- (generate);
    \draw[->] (generate) -- (verify);
    \draw[->] (verify) -- (evaluate);
    \draw[->] (evaluate) -- (feedback);
    \draw[->] (feedback) -- (select);
    \draw[->, dashed] (select) --
        node[lbl, right, align=left]{budget exhausted\\/ converged} (final);

    \draw[->, thick, violet!60!black]
        (select) to[out=15, in=-15, looseness=1.8]
        node[lbl, pos=0.5, left, align=left, text=violet!60!black]
        {parents + witnesses\\(next generation)}
        (generate);
    \end{tikzpicture}
    \caption{Verification-guided discovery loop. Candidate generation
    (violet) is untrusted; deterministic verification (red) and
    evaluation (green) determine promotion into the archive (orange).
    Elite and structurally diverse archive members are returned, together
    with structured failure witnesses, to generate the next round of
    candidates. The loop terminates when the search budget is exhausted
    or performance converges, yielding the validation-selected final
    candidate.}
    \label{fig:vera_loop}
\end{figure}

\subsection{Discovery Scenario Partitioning}
\label{subsec:scenario_partition}

To separate algorithm discovery from final performance assessment, the
available scenario set is partitioned as
\begin{equation}
    \mathcal{S}
    =
    \mathcal{S}_{\mathrm{train}}
    \cup
    \mathcal{S}_{\mathrm{val}}
    \cup
    \mathcal{S}_{\mathrm{test}},
    \label{eq:scenario_partition}
\end{equation}
with pairwise-disjoint subsets. A scenario $s\in\mathcal{S}$ contains
a topology, link state, transport-request set, active intent, and, where
applicable, failure or telemetry perturbations.

The evolutionary loop receives performance feedback only from
$\mathcal{S}_{\mathrm{train}}$. The validation set
$\mathcal{S}_{\mathrm{val}}$ is used to select the final algorithm from
the discovery archive, while $\mathcal{S}_{\mathrm{test}}$ remains
unseen until the final experimental evaluation. In cross-topology
experiments, complete topology families rather than only traffic seeds
are held out, preventing the LLM from adapting its algorithm to the
structure of the test graph.

We further define
\begin{equation}
    \mathcal{S}^{\mathrm{oracle}}
    \subseteq
    \mathcal{S}
\end{equation}
as the subset of scenarios for which the reference optimization problem
can be solved to certified optimality within the prescribed solver
budget. Oracle results from
$\mathcal{S}^{\mathrm{oracle}}_{\mathrm{train}}$ may be used to generate
discovery feedback, whereas oracle results from held-out test scenarios
are used only for final reporting.

\subsection{LLM-Based Candidate Generation}
\label{subsec:llm_generation}

Let $\mathcal{P}^{(g)}$ denote the population of candidate algorithms at
generation $g$. Each candidate
\begin{equation}
    A_i^{(g)}
    =
    \left(
    h_{\mathrm{ord},i}^{(g)},
    h_{\mathrm{path},i}^{(g)}
    \right)
\end{equation}
is represented as a typed expression tree conforming to the domain-
specific language in \eqref{eq:dsl_operators}.

The LLM receives a constrained generation context
\begin{equation}
    \mathcal{C}_{\mathrm{LLM}}
    =
    \left(
    \mathcal{I}_{\mathrm{IR}},
    \mathcal{G}_{\mathrm{DSL}},
    A_{\mathrm{parent}},
    \mathbf{M}_{\mathrm{parent}},
    \mathcal{W}_{\mathrm{parent}}
    \right),
    \label{eq:llm_context}
\end{equation}
where $\mathcal{G}_{\mathrm{DSL}}$ is the permitted grammar,
$A_{\mathrm{parent}}$ contains one or more parent algorithms,
$\mathbf{M}_{\mathrm{parent}}$ summarizes their measured performance,
and $\mathcal{W}_{\mathrm{parent}}$ contains selected failure witnesses
defined in Section~\ref{subsec:failure_witness}.

The LLM is instructed to return a structured abstract syntax tree (AST)
rather than unrestricted executable source code. A candidate output has
the conceptual structure
\begin{equation}
\begin{split}
A =
\{&
\texttt{request\_ordering}: \mathcal{T}_{\mathrm{ord}},\\
&
\texttt{path\_ranking}: \mathcal{T}_{\mathrm{path}}
\},
\end{split}
\label{eq:generated_ast}
\end{equation}
where $\mathcal{T}_{\mathrm{ord}}$ and
$\mathcal{T}_{\mathrm{path}}$ are typed expression trees composed only
of operators and features allowed by
$\mathcal{I}_{\mathrm{IR}}$.

Candidate generation supports the following bounded transformations:
\begin{enumerate}
    \item \emph{parameter mutation}: modification of constants,
    coefficients, or thresholds;
    \item \emph{feature mutation}: insertion, removal, or replacement
    of permitted request/path features;
    \item \emph{operator mutation}: replacement of an arithmetic,
    comparison, aggregation, or conditional operator;
    \item \emph{branch mutation}: addition, removal, or modification of
    a bounded conditional branch;
    \item \emph{subtree recombination}: combination of compatible
    subexpressions from two parent algorithms; and
    \item \emph{repair}: targeted modification of an algorithm in
    response to a structured failure witness.
\end{enumerate}

This representation allows structural algorithm variation while
preventing the LLM from generating arbitrary network-controller code.
The originating LLM is not required when a candidate is subsequently
executed because the accepted AST is serialized and replayed
deterministically.

\subsection{Structural and Runtime Verification}
\label{subsec:program_verification}

Every generated candidate is verified before it is exposed to a network
scenario. We distinguish structural verification from network
feasibility preservation.

Let
\begin{equation}
    \mathcal{V}_{\mathrm{stat}}(A)\in\{0,1\}
\end{equation}
denote the static program verifier. A candidate is accepted by
$\mathcal{V}_{\mathrm{stat}}$ only if all of the following conditions
hold:
\begin{enumerate}
    \item the AST conforms to the declared schema and contains only
    whitelisted DSL operators;
    \item all referenced features belong to $\mathcal{F}$;
    \item expression inputs and outputs satisfy their declared types;
    \item constants and thresholds lie inside predefined numerical
    bounds;
    \item the AST depth and total number of nodes do not exceed the
    configured complexity limits; and
    \item no prohibited capabilities, including external calls, file
    access, mutable network state, recursion, or unbounded iteration,
    are present.
\end{enumerate}

For structurally valid candidates, a trusted compiler converts the AST
into an executable scoring function. The resulting program is evaluated
inside a restricted runtime environment. Define
\begin{equation}
    \mathcal{V}_{\mathrm{run}}(A)\in\{0,1\}
\end{equation}
as the runtime verifier. It evaluates the candidate on a set
$\mathcal{S}_{\mathrm{micro}}$ of adversarial micro-instances designed
to expose invalid numerical or computational behavior:
\begin{equation}
\begin{aligned}
\mathcal{V}_{\mathrm{run}}(A)=1\iff{}&
\text{outputs are finite and deterministic},\\
&\text{with no exception or budget violation}.
\end{aligned}
\label{eq:runtime_verifier}
\end{equation}

A candidate is executable only if
\begin{equation}
    \mathcal{V}_{\mathrm{prog}}(A)
    =
    \mathcal{V}_{\mathrm{stat}}(A)
    \mathcal{V}_{\mathrm{run}}(A)
    =
    1.
    \label{eq:program_acceptance}
\end{equation}

Program verification and network feasibility serve different purposes.
Equation~\eqref{eq:program_acceptance} determines whether generated
decision logic is safe to execute, whereas
Proposition~\ref{prop:feasibility} guarantees that a program executed
through the trusted allocator cannot directly violate the modeled
capacity, latency, path-validity, or single-path constraints.

\begin{table}[t]
\centering
\caption{Verification and Evaluation Stages in VERA-TN}
\label{tab:verification_stages}
\scriptsize
\begin{tabular}{p{0.18\columnwidth}
                p{0.27\columnwidth}
                p{0.37\columnwidth}}
\hline
\textbf{Stage} &
\textbf{Purpose} &
\textbf{Typical feedback} \\
\hline
Static verifier &
Check DSL/schema/type and complexity rules &
Invalid operator, undeclared feature, excessive depth \\
Runtime verifier &
Detect numerical and execution failures &
NaN/Inf, timeout, exception, nondeterminism \\
Trusted allocator &
Preserve network feasibility &
Rejected path/request and responsible resource \\
Scenario evaluator &
Measure intent performance &
Admission loss, congestion, reconfiguration, runtime \\
Exact oracle &
Assess solution quality on tractable cases &
Optimality gap and differing admission/path decisions \\
NDT check &
Assess robustness &
Failure sensitivity, SLA violation, recovery cost \\
\hline
\end{tabular}
\end{table}

\subsection{Scenario-Based Performance Evaluation}
\label{subsec:candidate_evaluation}

For every executable candidate $A$ and discovery scenario
$s\in\mathcal{S}_{\mathrm{train}}$, the trusted allocator produces
\begin{equation}
    X_{A,s}
    =
    \mathcal{T}(A,s),
    \label{eq:trusted_execution}
\end{equation}
where $X_{A,s}$ contains admission decisions, selected paths, resulting
link states, and an execution trace recording the ordering and ranking
decisions made by $A$.

The normalized priority-weighted admission utility is
\begin{equation}
    U_{\mathrm{adm}}(A,s)
    =
    \frac{
        \sum_{k\in\mathcal{R}_s}\pi_k a_k^{A,s}
    }{
        \sum_{k\in\mathcal{R}_s}\pi_k
    }.
    \label{eq:admission_utility}
\end{equation}

Peak post-allocation utilization is
\begin{equation}
    U_{\max}(A,s)
    =
    \max_{e\in\mathcal{E}_s}
    \frac{l_e^{A,s}}{C_e}.
    \label{eq:evaluation_max_utilization}
\end{equation}

When a scenario contains an existing allocation, reconfiguration cost is
represented as
\begin{equation}
    C_{\mathrm{reconf}}(A,s)
    =
    \sum_{k\in\mathcal{R}_s}
    \mathbf{1}
    \left[
    p_k^{A,s}
    \neq
    p_k^{\mathrm{prev}}
    \right],
    \label{eq:reconfiguration_metric}
\end{equation}
where the indicator is evaluated only for previously admitted services.

For dynamic or network-digital-twin scenarios, the realized SLA
violation ratio is
\begin{equation}
    V_{\mathrm{SLA}}(A,s)
    =
    \frac{
        \sum_{k\in\mathcal{R}_s}
        \mathbf{1}
        \left[
        D_{k}^{\mathrm{realized}}
        >
        L_k^{\max}
        \right]
    }{
        \max(1,|\mathcal{R}_s^{\mathrm{adm}}|)
    }.
    \label{eq:sla_violation_ratio}
\end{equation}

The analytical model does not use
$V_{\mathrm{SLA}}$ as an independent feasibility metric because
latency-infeasible paths have already been removed. Instead,
\eqref{eq:sla_violation_ratio} captures model mismatch and dynamic
behavior observed in an emulated or digital-twin environment.

Candidate execution time is denoted by
$T_{\mathrm{exec}}(A,s)$. The scenario-level metric vector is therefore
\begin{equation}
\mathbf{m}(A,s)
=
\left[
U_{\mathrm{adm}},
-
V_{\mathrm{SLA}},
-
U_{\max},
-
C_{\mathrm{reconf}},
-
T_{\mathrm{exec}}
\right],
\label{eq:scenario_metric_vector}
\end{equation}
with inactive metrics omitted when they are not specified by the
corresponding intent.

Rather than imposing a universal weighted sum, VERA-TN derives the
candidate ordering from the objective hierarchy $\mathcal{O}$ in
$\mathcal{I}_{\mathrm{IR}}$. For a scenario set
$\mathcal{S}'$, let
\begin{equation}
    \mathbf{M}(A;\mathcal{S}')
    =
    \operatorname{Aggregate}_{s\in\mathcal{S}'}
    \mathbf{m}(A,s)
    \label{eq:aggregate_metric}
\end{equation}
denote the corresponding aggregate metric vector. Candidate algorithms
are compared using the lexicographic or Pareto ordering specified by
$\mathcal{O}$, while program validity always dominates performance.

\subsection{Exact-Oracle Comparison}
\label{subsec:oracle_feedback}

For scenarios in $\mathcal{S}^{\mathrm{oracle}}$, the optimization
problem in Section~\ref{subsec:reference_optimization} provides an
external quality reference independent of the LLM.

For scenario $s$, let $J^\star_s$ denote the certified optimal objective
and $J(A,s)$ the corresponding objective obtained by candidate $A$.
The candidate's relative optimality loss is
\begin{equation}
    \mathrm{Gap}(A,s)
    =
    \frac{
        J^\star_s-J(A,s)
    }{
        \max(|J^\star_s|,\epsilon)
    }.
    \label{eq:scenario_optimality_gap}
\end{equation}

The oracle is used for two purposes. First, it provides an objective
benchmark for evaluating whether discovered algorithms approach the
quality of exact optimization. Second, on training instances only, the
difference between the candidate and oracle allocation can be converted
into structured feedback identifying decisions responsible for the
quality loss.

If the solver terminates without a certificate of optimality, its
incumbent solution is not denoted by $J^\star_s$. Instead, the incumbent
and solver bound are recorded separately, and the corresponding
instance is excluded from exact optimality-gap claims.

\subsection{Structured Failure Witnesses}
\label{subsec:failure_witness}

A scalar objective reveals whether a candidate performs poorly but
provides little information about how its decision logic should change.
VERA-TN therefore converts verification and evaluation failures into a
structured failure witness
\begin{equation}
    w
    =
    \left(
    s,
    \zeta,
    \mathcal{R}^{w},
    \mathcal{E}^{w},
    y_{\mathrm{obs}},
    y_{\mathrm{ref}},
    \xi
    \right),
    \label{eq:failure_witness}
\end{equation}
where $s$ identifies the scenario, $\zeta$ is the failure type,
$\mathcal{R}^{w}$ and $\mathcal{E}^{w}$ contain affected requests and
links, $y_{\mathrm{obs}}$ is the observed candidate behavior,
$y_{\mathrm{ref}}$ is a reference behavior when available, and $\xi$
is a compact deterministic explanation generated by the evaluator.

Witnesses can originate from several sources.

\emph{Program witness:} a candidate violates the DSL, produces a
non-finite score, times out, or raises an exception.

\emph{Admission-regret witness:} a high-priority request is rejected
while an oracle or stronger reference allocation admits it. For an
oracle-solvable scenario, request-level admission regret is
\begin{equation}
    \Delta_k^{\mathrm{adm}}
    =
    a_k^\star-a_k^{A},
    \label{eq:admission_regret}
\end{equation}
and requests with
$\Delta_k^{\mathrm{adm}}=1$ are traced back to the ordering and
path-ranking decisions that exhausted the required resources.

\emph{Congestion witness:} the candidate creates a substantially larger
peak utilization than a reference solution. The most critical link is
\begin{equation}
    e^\dagger
    =
    \arg\max_{e\in\mathcal{E}}
    \frac{l_e^A}{C_e}.
    \label{eq:critical_link}
\end{equation}
The witness includes the requests routed through $e^\dagger$ and the
alternative paths that were available when those decisions were made.

\emph{Robustness witness:} a failure, demand shift, or telemetry
perturbation causes an excessive loss in admission, SLA compliance, or
recovery performance.

\emph{Reconfiguration witness:} the candidate achieves only a marginal
utility improvement while causing excessive path churn.

The evaluator ranks witnesses according to their effect on the current
intent hierarchy and retains at most $N_w$ representative witnesses per
parent. This limits prompt length and prevents large scenario traces from
being passed directly to the LLM.

We use the term \emph{failure witness} rather than claiming a formal
counterexample in the program-verification sense. These witnesses guide
evolution by identifying observed deficiencies over the evaluated
scenario set; they do not constitute a proof that the candidate fails
for every network state.

\subsection{Verification-Guided Repair and Evolution}
\label{subsec:evolution}

At generation $g$, the evaluator first removes candidates for which
$\mathcal{V}_{\mathrm{prog}}(A)=0$. Remaining candidates are ranked
according to their training performance under the intent-defined
ordering.

Let
\begin{equation}
    \mathcal{E}^{(g)}
    \subseteq
    \mathcal{P}^{(g)}
\end{equation}
denote the elite subset. Parent algorithms are selected from both
high-performing and structurally diverse candidates. To prevent the
population from collapsing to nearly identical expression trees, we
associate each candidate with a structural signature consisting of the
set of features, operators, thresholds, and branch structures appearing
in its AST.

For two candidate algorithms $A_i$ and $A_j$, a normalized structural
diversity measure can be expressed as
\begin{equation}
d_{\mathrm{str}}(A_i,A_j)
=
\lambda
d_{\mathrm{tree}}(A_i,A_j)
+
(1-\lambda)
\left[
1-
\frac{
|\Psi_i\cap\Psi_j|
}{
|\Psi_i\cup\Psi_j|
}
\right],
\label{eq:structural_diversity}
\end{equation}
where $d_{\mathrm{tree}}$ is normalized tree-edit distance,
$\Psi_i$ and $\Psi_j$ are the sets of operators and features used by the
two candidates, and $\lambda\in[0,1]$ controls the relative contribution
of tree structure and feature/operator overlap.

The next candidate population is generated by
\begin{equation}
    \mathcal{P}^{(g+1)}
    =
    \operatorname{Select}
    \left(
    \mathcal{E}^{(g)}
    \cup
    \mathcal{C}^{(g)};
    \mathcal{O},
    d_{\mathrm{str}}
    \right),
    \label{eq:survivor_selection}
\end{equation}
where $\mathcal{C}^{(g)}$ denotes the set of newly generated children.
Selection first enforces program validity, then ranks candidates by the
intent-defined objective ordering, while reserving part of the
population for structurally distinct candidates above a minimum
diversity threshold.

For repair-oriented generation, the LLM receives the parent algorithm
together with a bounded set of witnesses and is explicitly instructed
to modify only the decision logic relevant to the observed deficiency.
For example, if a high-priority request is rejected because earlier
low-priority requests consume a bottleneck link, the feedback identifies
the affected requests, the bottleneck, and the parent ranking behavior.
The LLM may respond by changing request ordering, introducing a
conditional rule for tight resources, or modifying path ranking to
preserve bottleneck capacity.

This process is \emph{counterexample-guided in spirit}, but we do not
claim equivalence to formal counterexample-guided inductive synthesis
(CEGIS). The feedback is obtained from finite network scenarios and
optimization/digital-twin evaluation rather than from exhaustive
symbolic verification of the complete program state space.

\begin{algorithm}[t]
\caption{Verification-Guided LLM Algorithm Evolution}
\label{alg:vera}
\begin{algorithmic}[1]
\Require Intent IR $\mathcal{I}_{\mathrm{IR}}$,
training scenarios $\mathcal{S}_{\mathrm{train}}$,
validation scenarios $\mathcal{S}_{\mathrm{val}}$,
population size $N$,
generation limit $G_{\max}$,
evaluation budget $B_{\max}$
\Ensure Selected transport-allocation algorithm $A^\star$

\State Initialize seed population $\mathcal{P}^{(0)}$
\State $\mathcal{H}\gets\emptyset$ \Comment{archive}
\State $B\gets 0$

\For{$g=0,\ldots,G_{\max}-1$}

    \ForAll{$A\in\mathcal{P}^{(g)}$}

        \If{$\mathcal{V}_{\mathrm{stat}}(A)=0$}
            \State Record structural witness for $A$
            \State \textbf{continue}
        \EndIf

        \If{$\mathcal{V}_{\mathrm{run}}(A)=0$}
            \State Record runtime witness for $A$
            \State \textbf{continue}
        \EndIf

        \State Execute $A$ using trusted allocator on
        $\mathcal{S}_{\mathrm{train}}$
        \State Compute $\mathbf{M}(A;\mathcal{S}_{\mathrm{train}})$
        \State Compare with exact oracle where available
        \State Extract bounded witness set $\mathcal{W}_A$
        \State Add $(A,\mathbf{M},\mathcal{W}_A)$ to archive $\mathcal{H}$
        \State $B\gets B+|\mathcal{S}_{\mathrm{train}}|$

        \If{$B\geq B_{\max}$}
            \State \textbf{break}
        \EndIf

    \EndFor

    \If{$B\geq B_{\max}$}
        \State \textbf{break}
    \EndIf

    \State Select elite and diversity-preserving parents
    \State $\mathcal{C}^{(g)}\gets\emptyset$

    \While{$|\mathcal{C}^{(g)}|<N$}
        \State Select parent(s) $A_p$ from $\mathcal{P}^{(g)}$
        \State Build constrained LLM context using
        $\mathcal{I}_{\mathrm{IR}}$, $A_p$, metrics, and witnesses
        \State Generate mutation, recombination, or repair $A_c$
        \State Deduplicate using canonical AST hash
        \State $\mathcal{C}^{(g)}\gets\mathcal{C}^{(g)}\cup\{A_c\}$
    \EndWhile

    \State Form $\mathcal{P}^{(g+1)}$ according to
    \eqref{eq:survivor_selection}

\EndFor

\State Evaluate non-dominated archive candidates on
$\mathcal{S}_{\mathrm{val}}$
\State Select $A^\star$ using the intent-defined validation ordering
\State \Return $A^\star$
\end{algorithmic}
\end{algorithm}

\subsection{Initialization, Termination, and Final Selection}
\label{subsec:termination}

The initial population combines human-designed seeds and LLM-generated
candidates. Human seeds include shortest-path, latency-aware,
bandwidth-aware, and load-balancing heuristics. A restricted linear
candidate of the form
\begin{equation}
h(\mathbf{f})
=
\boldsymbol{\theta}^{\mathsf{T}}\mathbf{f}
\label{eq:linear_seed}
\end{equation}
is also included. This candidate is important as an ablation because it
isolates coefficient optimization from structural algorithm discovery.

Discovery terminates when any of the following conditions is reached:
\begin{enumerate}
    \item the maximum generation count $G_{\max}$ is reached;
    \item the candidate-evaluation budget $B_{\max}$ is exhausted; or
    \item no validation-eligible improvement is observed for
    $G_{\mathrm{stall}}$ consecutive generations.
\end{enumerate}

The test set is not used to determine the stopping point. After
evolution, a compact set of non-dominated candidates from the archive
$\mathcal{H}$ is evaluated on
$\mathcal{S}_{\mathrm{val}}$. The final algorithm is selected as
\begin{equation}
    A^\star
    =
    \arg\max_{A\in\mathcal{H}_{\mathrm{valid}}}
    F
    \left(
    A
    \mid
    \mathcal{I},
    \mathcal{S}_{\mathrm{val}}
    \right).
    \label{eq:final_algorithm_selection}
\end{equation}

Only after $A^\star$ has been fixed is it evaluated on
$\mathcal{S}_{\mathrm{test}}$. The LLM therefore receives neither the
test scenarios nor their performance traces during discovery.

\subsection{Computational Complexity and Deployment Cost}
\label{subsec:complexity}

Let $R=|\mathcal{R}|$ denote the number of transport requests,
$K$ the maximum number of candidate paths per request, and $H$ the
maximum number of links in a candidate path. For an already discovered
algorithm, request sorting requires
$\mathcal{O}(R\log R)$ operations. Ranking at most $K$ candidate paths
for every request requires
$\mathcal{O}(RK\log K)$ score comparisons, while feasibility checking
over at most $H$ links per candidate contributes
$\mathcal{O}(RKH)$.

The resulting trusted online-allocation complexity is therefore
\begin{equation}
\mathcal{O}
\left(
R\log R
+
RK\log K
+
RKH
\right).
\label{eq:online_complexity}
\end{equation}

Importantly, this online execution does not require an LLM call. LLM
inference is confined to the offline or near-line discovery process.

Let $N$ denote population size, $G$ the number of generations, and
$S_{\mathrm{tr}}=|\mathcal{S}_{\mathrm{train}}|$. Excluding exact-oracle
computation, discovery requires at most
\begin{equation}
    \mathcal{O}
    \left(
    NGS_{\mathrm{tr}}
    \left[
    R\log R+
    RK\log K+
    RKH
    \right]
    \right)
    \label{eq:discovery_complexity}
\end{equation}
candidate-scenario evaluation work, in addition to the LLM-generation
calls.

Exact optimization is intentionally excluded from the online complexity
in \eqref{eq:online_complexity}. The reference problem includes
unsplittable path admission and can become combinatorial with network
size; it is therefore used as a discovery/evaluation oracle only where
a certified solution can be obtained within the configured solver
budget.

\subsection{Reproducibility of the Discovery Process}
\label{subsec:discovery_reproducibility}

For every LLM-generated candidate, VERA-TN records the model identifier
and version, generation parameters, prompt-template version, parent
candidate hashes, generation index, mutation type, serialized AST,
structural-verification result, runtime-verification result, scenario
metrics, failure witnesses, token usage, wall-clock generation time,
and random seed where applicable.

Accepted and rejected candidates are both retained. This is important
because reporting only successful programs would obscure the reliability
and computational cost of LLM-based discovery. The complete archive
therefore enables measurement of
\begin{equation}
\begin{aligned}
\mathcal{M}_{\mathrm{repro}}=\{&\text{valid-program and runtime-success rates},\\
&\text{time to first/best algorithm},\\
&\text{evaluations to convergence},\\
&\text{LLM token and compute cost}\}.
\end{aligned}
\label{eq:discovery_repro_metrics}
\end{equation}

Because the executable representation of every accepted candidate is
stored independently of the originating language model, final transport
allocation experiments can be deterministically replayed without
re-querying the LLM. This separation permits the stochasticity of
algorithm discovery to be analyzed independently from the runtime
behavior of the discovered network algorithm.

\section{Implementation and Reproducibility}
\label{sec:prototype}

This section separates the reference architecture in
Sections~\ref{sec:representation}--\ref{sec:verified_evolution} from the
released proof-of-concept. The architecture specifies typed AST generation,
static checks, sandboxed execution, and model-connected repair. The current
artifact implements the trusted allocator, exact oracle, deterministic
scenario generator, bounded candidate evaluator, numerical evolutionary
search, perturbation tests, and result archive. It does not execute arbitrary
model-generated Python, issue live LLM calls, or establish a multi-provider
typed-AST comparison. This distinction keeps the implementation claim aligned
with the released code while preserving the full design contract for the
next model-connected stage.

\subsection{Topology and Workload}
\label{subsec:topology}
TEFNET24 motivates the hierarchical structure of the representative graph
\cite{tefnet24}. The public reference describes a multi-layer packet-optical
network inspired by operational Telef\'onica deployments; its data interface
is available separately \cite{tefnetdata}. The present artifact instead uses
the deterministic 28-node, 54-link graph in Fig.~\ref{fig:topology}, which
preserves access aggregation, metro rings, redundant interconnection, dual
homing, and a connected national/core layer. Capacities and delays are
generated once per seed from tier-specific ranges. 

\begin{figure}[htbp]
    \centering
    \includegraphics[width=.5\textwidth]{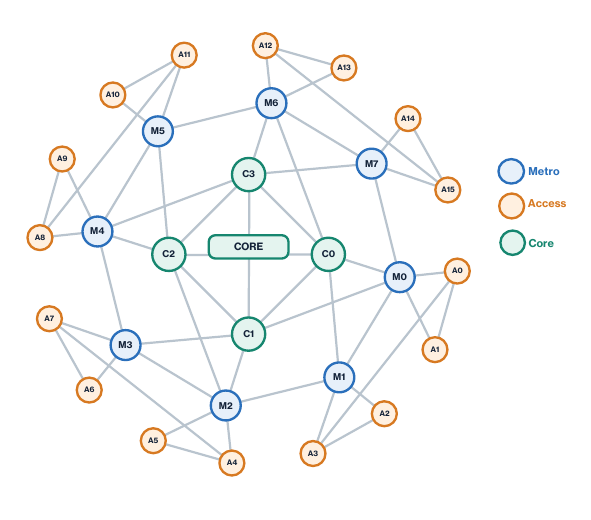}
    \caption{Representative TEFNET24-derived hierarchy used by the released
    proof-of-concept with 28 nodes, 54 bidirectional links.
 Capacities, delays, and workloads are deterministic
    experimental inputs rather than operator measurements.}
    \label{fig:topology}
\end{figure}

Each request contains ingress, egress, bandwidth, latency bound, and an
integer service priority. Up to four simple paths are generated
deterministically and paths exceeding the latency bound are removed before
ranking. Evaluation loads scale nominal demands by factors
$\{1.0, 1.4, 1.8, 2.2, 2.5\}$. Perturbed scenarios remove links and vary capacities
using archived seeds.

\subsection{Released Prototype}
\label{subsec:prototype_architecture}
The numerical candidate is a ten-parameter vector controlling request
priority/order and path features including projected utilization, latency,
residual slack, and risk. Evolution uses a population of 24 candidates for
14 generations. A performance-only search is evaluated on nominal scenarios;
the VERA-TN prototype additionally uses perturbed scenarios and an overload
penalty. Both execute through the same trusted allocator, so candidate logic
returns scores but never commits resource state.

The exact reference is the path-based binary model in
\eqref{eq:reference_objective}--\eqref{eq:binary_variables}, solved with
SciPy's MILP interface. The implementation records solver success, status,
objective, message, and MIP gap. A row is called oracle-optimal only when the
solver reports success and zero recorded gap within numerical tolerance.

\begin{table}[!tb]
\centering
\caption{Released Proof-of-Concept Configuration}
\label{tab:prototype_setup}
\scriptsize
\begin{tabular}{p{0.29\columnwidth}p{0.58\columnwidth}}
\toprule
\textbf{Item} & \textbf{Configuration} \\
\midrule
Topology & 28 nodes, 54 bidirectional links \\
Candidate paths & Up to $K=4$, loop-free and latency-feasible \\
Candidate representation & Bounded ten-parameter numerical vector \\
Evolution & Population 24; 14 generations \\
Oracle & Path-based binary MILP; recorded certificate status \\
Validation & Nominal plus analytical link/capacity perturbations \\
Evaluation seeds & 30 per load; 30 failed-link cases; 10 for K sensitivity \\
Online LLM inference & None \\
\bottomrule
\end{tabular}
\end{table}


\section{Evaluation}
\label{sec:evaluation}

\subsection{Scope, Questions, and Metrics}
\label{subsec:research_questions}
The study asks seven questions. \textbf{RQ1} measures priority utility under
increasing contention. \textbf{RQ2} tests whether evolution improves on
priority-greedy routing and equal-budget random search. \textbf{RQ3}
compares congestion with MILP-C, which minimizes the same metric at maximum
priority utility. \textbf{RQ4} tests failure-aware search on held-out
failed-link cases. \textbf{RQ5} tests whether three training intents yield
specialized candidates. \textbf{RQ6} tests whether intent behavior transfers
from the official national topology to unseen metro-regional topologies. The synthetic intent-transfer test uses resilience as its third intent, whereas the official-data study uses admitted volume. The two transfer tests are therefore interpreted separately.
\textbf{RQ7} measures $K$ sensitivity and separates online heuristic scaling
from exact-solver tractability.

The primary metric is priority utility divided by MILP-U utility. Secondary
metrics are peak link utilization, accepted requests, the fraction of cases
that reach exact priority utility, rejected priority under failure, and
runtime. All methods use the same path sets and trusted allocator. Capacity
and latency violations are therefore checked as invariants, not treated as
performance metrics.

\subsection{Baselines and Statistical Protocol}

Shortest path uses arrival order and the first feasible path. Load-balanced
KSP ranks paths with a fixed bottleneck-oriented score. Priority-greedy
orders requests by descending priority, breaks ties by ascending bandwidth,
and chooses the path with minimum projected bottleneck utilization.
Equal-budget random search controls for the number of sampled parameter
vectors. Evolutionary search uses the same candidate budget. MILP-U and
MILP-C are offline references, not deployment baselines.

Table values aggregate 150 paired test cases. Mean utility and utilization
use 10,000-resample bootstrap 95\% confidence intervals. Runtime is reported
as median and interquartile range (IQR). Pre-specified paired Wilcoxon tests
compare evolution and random search with priority-greedy, and evolution with
random search. Holm correction is applied across the six utility and
utilization tests. Paired rank-biserial correlation is the effect size.
For the official-data study, each of eight discovery seeds is the inferential
unit. The 16 national or 48 metro scores are first averaged within a seed.
Paired Wilcoxon tests then compare each intent-matched candidate with the two
off-diagonal candidates. Holm correction is applied across six contrasts per
domain, and bootstrap intervals resample the eight seed-level means. This
avoids treating repeated replay scenarios as independent discovery runs.

\subsection{Motivating Toy Example}
\label{subsec:toy_example}
Before presenting aggregate results, we use a small constructed example
to illustrate the mechanism through which request ordering and path
ranking can affect priority utility. Both policies in
Fig.~\ref{fig:toy_routes} operate
on the same topology and candidate paths and are subject to exactly the
same capacity and latency checks. The only difference is the order in
which requests and feasible paths are considered.
Under arrival-order shortest-path allocation, early routing decisions
consume resources that could otherwise support higher-value demand,
yielding total priority utility $6$. The bounded VERA-TN prototype
changes the ordering/ranking decisions and reaches utility $8$, equal
to the reference MILP utility for this constructed instance. Thus, the
illustration isolates the mechanism targeted by algorithm discovery:
the gain arises from how scarce feasible resources are prioritized, not
from relaxing network constraints.


\begin{figure}[!htb]
    \centering
    \includegraphics[width=.493\textwidth]{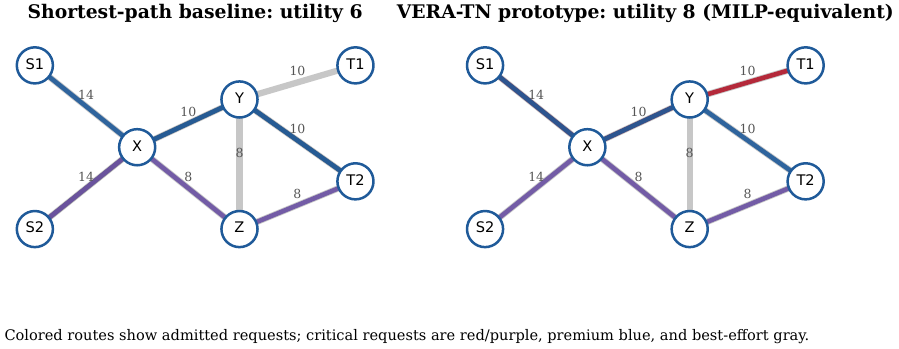}
    \caption{Constructed contention illustration. The example shows why
    request ordering and path ranking matter; it is not part of the
    statistical evaluation.}
    \label{fig:toy_routes}
\end{figure}

\subsection{Priority Utility and Congestion}

Table~\ref{tab:main_results} and Fig.~\ref{fig:evaluation} answer
RQ1--RQ3. The higher demand levels avoid the ceiling effect observed in the
earlier evaluation.
Evolutionary search reaches mean utility ratio 0.958, compared with 0.940
for priority-greedy and 0.952 for equal-budget random search. Its mean
paired gain over priority-greedy is 0.0187 (Holm-adjusted
$p=1.3\times10^{-6}$, rank-biserial $r=0.658$). The gain over random search
is 0.00635 ($p=0.0083$, $r=0.434$). The effect is statistically detectable
but small. Evolution reaches exact priority utility in 50 of 150 cases.

\begin{table*}[!tb]
\centering
\caption{Held-Out Results Across Five Loads and 30 Seeds per Load ($n=150$ per Method)}
\label{tab:main_results}
\scriptsize
\begin{tabular}{lccccc}
\toprule
\textbf{Method} & \textbf{Utility ratio, mean$\pm$SD [95\% CI]} &
\textbf{Peak util.} & \textbf{Accepted} &
\textbf{Runtime ms, median [IQR]} & \textbf{Priority-opt.} \\
\midrule
Shortest path & $0.805\pm0.115$ [0.787,0.823] & 0.974 & 16.71 & 0.198 [0.190,0.207] & 5.3\% \\
Load-balanced KSP & $0.830\pm0.120$ [0.811,0.849] & 0.955 & 17.27 & 1.564 [1.503,1.656] & 13.3\% \\
Priority-greedy & $0.940\pm0.050$ [0.932,0.947] & 0.948 & 17.51 & 0.350 [0.334,0.372] & 20.0\% \\
Equal-budget random & $0.952\pm0.046$ [0.945,0.959] & 0.963 & 17.79 & 1.583 [1.518,1.667] & 26.7\% \\
Evolutionary search & $0.958\pm0.041$ [0.952,0.965] & 0.961 & 17.95 & 1.573 [1.522,1.669] & 33.3\% \\
MILP-U & $1.000\pm0.000$ [1.000,1.000] & 0.977 & 19.03 & 26.358 [17.076,37.689] & 100\% \\
MILP-C & $1.000\pm0.000$ [1.000,1.000] & 0.939 & 19.03 & 84.957 [34.833,176.536] & 100\% \\
\bottomrule
\end{tabular}
\end{table*}

\begin{figure*}[!tb]
    \centering
    \includegraphics[width=.76\textwidth]{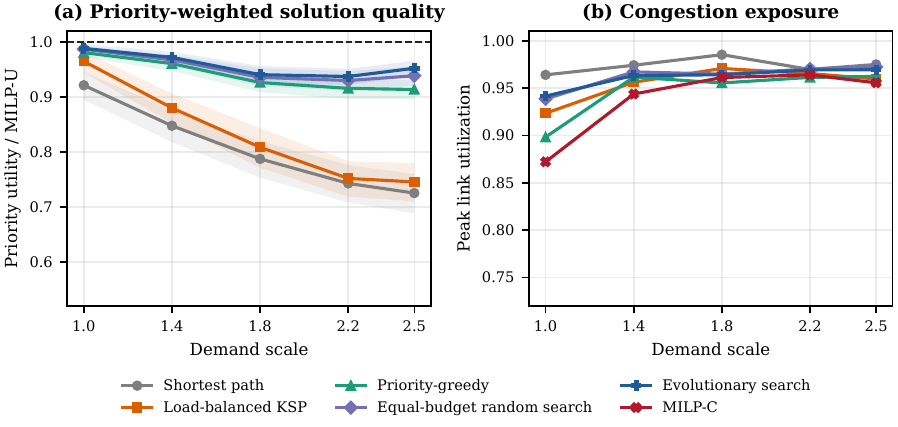}
    \caption{Held-out priority utility and peak utilization. Shaded bands
    in (a) are bootstrap 95\% confidence intervals. MILP-C is the fair
    congestion reference because it preserves MILP-U priority utility.}
    \label{fig:evaluation}
\end{figure*}

The heuristic does not improve congestion over MILP-C. MILP-C reaches mean
peak utilization 0.939; evolutionary search reaches 0.961. Moreover, the
candidate is priority-optimal in only one third of the cases. On those 50
cases, its mean utilization regret relative to MILP-C is 0.058. The earlier
latency-tie-break MILP is not used as a congestion reference.

\subsection{Failure Training and Intent Transfer}

On 30 held-out failed-link scenarios, failure-aware search reaches mean
utility ratio 0.934 and 10th-percentile ratio 0.862. Search without failure
scenarios reaches 0.919 and 0.844, respectively. The paired mean difference
is 0.0151 with bootstrap CI [0.0031,0.0288], but the Wilcoxon test gives
$p=0.051$. This test therefore does not establish a robustness gain. Removing
the utilization feature lowers the mean to 0.914. Removing the nearly zero
slack coefficient has no observed effect.

\begin{figure}[!htb]
    \centering
    \includegraphics[width=.78\columnwidth]{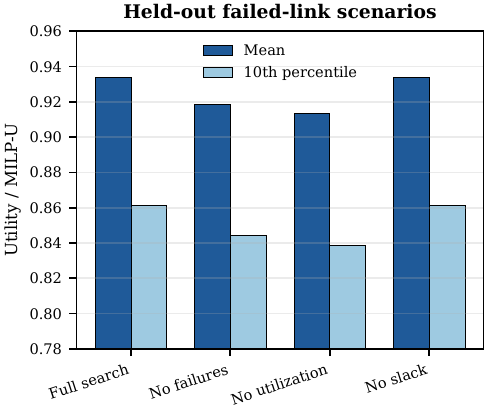}
    \caption{Mean and 10th-percentile utility over 30 held-out failed-link
    cases. Failure-aware search improves the observed mean and tail, but
    the paired Wilcoxon result is marginal.}
    \label{fig:ablation}
\end{figure}

The intent-transfer test shows no intent-specific specialization. Three searches use priority,
congestion, and resilience scores, but the priority-trained vector has the
best held-out score under all three objectives. Fig.~\ref{fig:intent_transfer}
reports regret from the best candidate for each target objective. Different
weights alone do not establish intent-conditioned specialization. The present
evaluation therefore does not support a claim of structural algorithm
discovery.

\begin{figure}[!htb]
    \centering
    \includegraphics[width=.78\columnwidth]{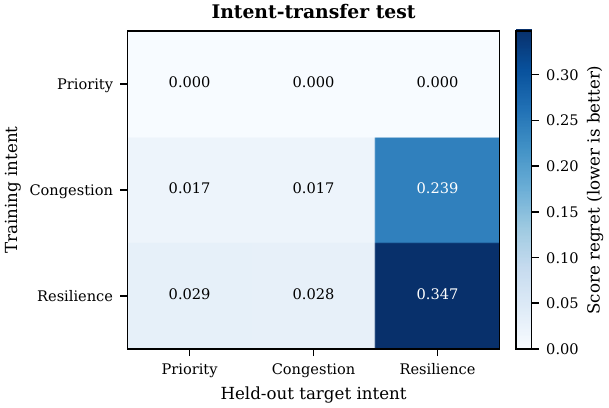}
    \caption{Intent-transfer score regret on 60 held-out nominal and failed
    cases. Zero is best. The diagonal does not dominate, so the numerical
    representation does not validate intent-specific specialization.}
    \label{fig:intent_transfer}
\end{figure}

\subsection{Official Cross-Topology Check}

The official-data study also shows no intent-specific specialization. Discovery on the national graph
does not produce a stable intent-matched diagonal. Across the three national
targets, the largest absolute matched contrast is below 0.0006 and no adjusted
test is significant.

Fig.~\ref{fig:cross_topology} replays all 24 selected candidates without
retraining on 48 workloads from 12 unseen metro components. The
congestion-trained candidates have the best mean score under all three metro
targets. None of the 12 intent-matched contrasts is significant; every
Holm-adjusted $p$-value equals 1.0. Priority-greedy obtains mean priority ratios
0.9785 nationally and 0.6126 on the metro cases, close to 0.9794 and 0.6134 for
the priority-trained candidates. All 1,536 candidate replays satisfy the
trusted capacity and latency checks.

\begin{figure*}[!tb]
    \centering
    \includegraphics[width=.8\textwidth]{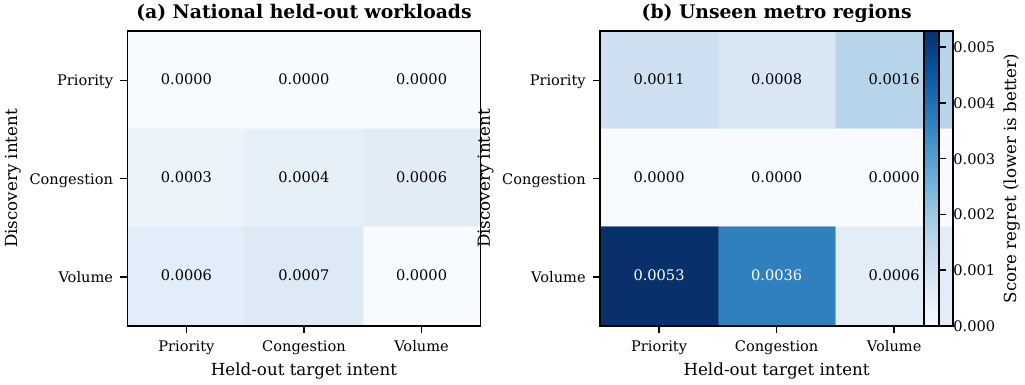}
    \caption{Official TEFNET24 intent-transfer check. Each cell is regret from
    the best mean score for its target objective. Discovery uses eight seeds
    on the national graph; metro evaluation uses 12 unseen components and 48
    workloads without retraining. The diagonal does not dominate in either
    domain. Inferential tests use discovery-seed means, not individual replay
    rows.}
    \label{fig:cross_topology}
\end{figure*}

\subsection{Sensitivity and Tractability}

Increasing $K$ raises median replay time from 0.53 ms at $K=1$ to 4.46 ms
at $K=16$. Mean utility does not improve monotonically: it is 0.974 at
$K=1$, 0.947 at $K=4$, and 0.934 at $K=16$ on the ten sensitivity seeds.
The fixed candidate was selected with $K=4$ and does not automatically use
additional paths well.

\begin{figure}[!htb]
    \centering
    \includegraphics[width=.98\columnwidth]{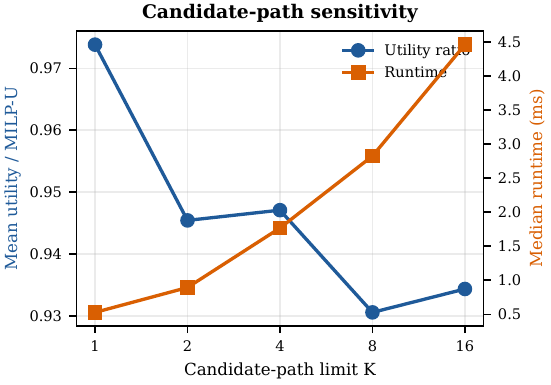}
    \caption{Candidate-path sensitivity over ten seeds. Larger path sets
    increase replay time but do not improve this fixed candidate
    monotonically.}
    \label{fig:k_sensitivity}
\end{figure}

\begin{figure*}[!tb]
    \centering
    \includegraphics[width=.76\textwidth]{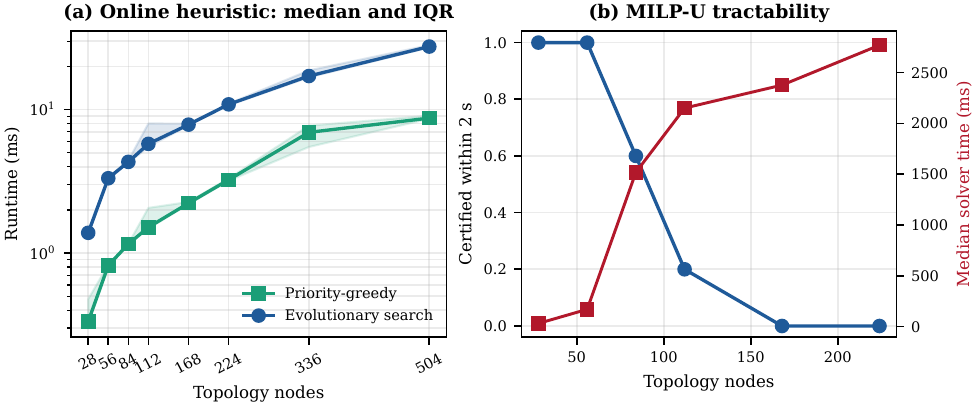}
    \caption{Runtime scaling. (a) Online replay uses 30 repetitions per
    point; bands show IQR. (b) MILP-U reports the certified fraction and
    median elapsed time for five cases per size under a two-second limit.}
    \label{fig:scaling}
\end{figure*}
Fig.~\ref{fig:scaling} separates online replay from exact optimization.
Across 28--504 nodes, fitted runtime exponents are 0.99 for evolutionary
replay and 1.14 for priority-greedy under this joint graph/request scaling.
These are empirical fits, not asymptotic proofs. With a two-second MILP-U
budget, all five cases are certified at 28 and 56 nodes, three of five at
84 nodes, one of five at 112 nodes, and none at 168 or 224 nodes. Uncertified
cases are excluded from Table~\ref{tab:main_results}; the separate curve
shows where the exact reference stops being dependable within that budget.


\section{Discussion}
\label{sec:discussion}

The evidence supports two claims. First, the trust boundary preserves the
modeled allocation constraints independently of the ranking parameters.
Second, under a fixed evaluation budget, numerical evolution gives a small
held-out utility gain over random sampling and a larger gain over the
priority-greedy baseline. The study does not support three stronger claims:
the evolved candidate is not congestion-optimal, failure training is not
significant at $\alpha=0.05$, and intent-specific candidates do not
specialize on held-out tests.

The study has several limitations. It uses batch allocation, analytical failures, one
parameterized representation, and one machine. The official-data check spans
national and metro graphs from one topology family; it does not establish
transfer to an independent operator or network family. Its link capacities are
assumed, and the public traffic matrix determines endpoint marginals rather
than per-request bandwidth. Concurrent arrivals, reconfiguration cost,
digital-twin fidelity, controller behavior, and physical-network effects are
not tested. The typed-AST and LLM extension is specified but not implemented.
A direct comparison with genetic programming on the same grammar is required
before making a structural-search claim.

The next empirical step should be a controlled structural study: implement the
typed grammar, compare GP
and at least two model generators under the same candidate budget, archive
all rejected programs, and evaluate on independent public topology families.
Dynamic multi-epoch tests are also needed before reporting reconfiguration
or operational resilience.

\section{Conclusion}
\label{sec:conclusion}
This paper introduced VERA-TN, a verification-guided architecture for turning
declarative transport intents into bounded, reusable request-ordering and
path-ranking procedures. The central safety property is architectural: a
candidate may rank prevalidated choices, while a trusted allocator retains
path construction, resource accounting, hard-constraint enforcement, and
actuation. Proposition~\ref{prop:feasibility} establishes allocation
feasibility under explicit assumptions, and Lemma~\ref{lem:lexicographic}
gives a sufficient bound for the exact model's latency tie-break.

Across 150 certified held-out cases, numerical evolution reaches mean priority
utility ratio 0.958, compared with 0.952 for equal-budget random search and
0.940 for priority-greedy routing. The gain over random search is small but
statistically detectable. The candidate does not match the congestion oracle,
and its intent-transfer test does not show specialization. Failure-aware search
improves the observed mean and tail, but its paired test is marginal. This
larger study separates statistically supported gains from null results.

Eight repeated searches on the official national graph and replay on 12 unseen
metro-regional graphs also show the limits of the ten-parameter
representation: no stable intent-matched diagonal appears.

The present contribution is the trust boundary and its reproducible numerical
evaluation. It is not evidence that an LLM improves heuristic design. That
question requires a typed structural implementation, equal-budget GP and model
controls, independent topology families, and operational validation.

\section*{Acknowledgment}
This work is partially funded by the Juan de la Cierva (JDC) Grant under reference (JDC2024-053571-I), the SHINE (PID2024-159781OB-I00) project funded by MICIU/AEI/10.13039/501100011033 and by ERDF/EU.


\scriptsize
\bibliographystyle{IEEEtran}
\bibliography{ref}

@article{funsearch,
  author={Romera-Paredes, Bernardino and Barekatain, Mohammadamin and Novikov, Alexander and Balog, Matej and Kumar, M. Pawan and Dupont, Emilien and Ruiz, Francisco J. R. and others},
  title={Mathematical Discoveries from Program Search with Large Language Models},
  journal={Nature},
  volume={625},
  pages={468--475},
  year={2024},
  doi={10.1038/s41586-023-06924-6}
}

@article{alphaevolve,
  author={Novikov, Alexander and V{\~u}, Ng{\^a}n and Eisenberger, Marvin and Dupont, Emilien and Huang, Po-Sen and Wagner, Adam Zsolt and Shirobokov, Sergey and Kozlovskii, Borislav and Ruiz, Francisco J. R. and Mehrabian, Abbas and others},
  title={{AlphaEvolve}: A Coding Agent for Scientific and Algorithmic Discovery},
  journal={arXiv preprint arXiv:2506.13131},
  year={2025}
}

@inproceedings{eoh,
  author={Liu, Fei and Tong, Xialiang and Yuan, Mingxuan and Lin, Xi and Luo, Fu and Wang, Zhenkun and Lu, Zhichao and Zhang, Qingfu},
  title={Evolution of Heuristics: Towards Efficient Automatic Algorithm Design Using Large Language Model},
  booktitle={Proc. 41st Int. Conf. Machine Learning (ICML)},
  volume={235},
  pages={32201--32223},
  year={2024}
}

@inproceedings{reevo,
  author={Ye, Haoran and Wang, Jiarui and Cao, Zhiguang and Berto, Federico and Hua, Chuanbo and Kim, Haeyeon and Park, Jinkyoo and Song, Guojie},
  title={{ReEvo}: Large Language Models as Hyper-Heuristics with Reflective Evolution},
  booktitle={Advances in Neural Information Processing Systems},
  volume={37},
  year={2024}
}

@article{llamea,
  author={van Stein, Niki and B{\"a}ck, Thomas},
  title={{LLaMEA}: A Large Language Model Evolutionary Algorithm for Automatically Generating Metaheuristics},
  journal={IEEE Transactions on Evolutionary Computation},
  volume={29},
  number={2},
  pages={331--345},
  year={2025},
  doi={10.1109/TEVC.2024.3497793}
}

@article{llm4ad,
  author={Liu, Fei and Zhang, Rui and Xie, Zhuoliang and Sun, Rui and Li, Kai and Lin, Xi and Wang, Zhenkun and Lu, Zhichao and Zhang, Qingfu},
  title={{LLM4AD}: A Platform for Algorithm Design with Large Language Models},
  journal={arXiv preprint arXiv:2412.17287},
  year={2024}
}

@article{llmrao,
  author={Noh, Hyeonho and Shim, Byonghyo and Yang, Hyun Jong},
  title={Adaptive Resource Allocation Optimization Using Large Language Models in Dynamic Wireless Environments},
  journal={IEEE Transactions on Vehicular Technology},
  volume={74},
  number={10},
  pages={16630--16635},
  year={2025},
  doi={10.1109/TVT.2025.3572440}
}

@inproceedings{netllm,
  author={Wu, Duo and Wang, Xianda and Qiao, Yaqi and Wang, Zhi and Jiang, Junchen and Cui, Shuguang and Wang, Fangxin},
  title={{NetLLM}: Adapting Large Language Models for Networking},
  booktitle={Proc. ACM SIGCOMM},
  pages={661--678},
  year={2024},
  doi={10.1145/3651890.3672268}
}

@article{tefnet24,
  author={Rivas-Moscoso, Jos{\'e} Manuel and Arpanaei, Farhad and Otero P{\'e}rez, Gabriel and Mart{\'i}nez Jim{\'e}nez, Jos{\'e} David and Fern{\'a}ndez-Palacios, Juan Pedro and Gonz{\'a}lez de Dios, {\'O}scar and Contreras, Luis Miguel and S{\'a}nchez-Maci{\'a}n, Alfonso and Hern{\'a}ndez, Jos{\'e} Alberto and Larrabeiti, David and Folgueira, Jes{\'u}s},
  title={{TEFNET24}: Reference Packet Optical Network Topology for Edge to Core Transport},
  journal={Journal of Optical Communications and Networking},
  volume={16},
  number={11},
  pages={G28--G39},
  year={2024},
  doi={10.1364/JOCN.533131}
}

@misc{tefnetdata,
  author={{Telef{\'o}nica Research and Development and Universidad Carlos III de Madrid}},
  title={{TEFNET24} Dataset},
  year={2024},
  publisher={Zenodo},
  doi={10.5281/zenodo.13362934}
}

@inproceedings{teraflow,
  author={Munoz, Ra{\'u}l and Vilalta, Ricard and Casellas, Ramon and others},
  title={{TeraFlow}: Secured Autonomic Traffic Management for a Tera of {SDN} Flows},
  booktitle={Proc. Joint European Conference on Networks and Communications and {6G} Summit},
  year={2021},
  doi={10.1109/EuCNC/6GSummit51104.2021.9482469}
}

@article{ibnsurvey,
  author={Leivadeas, Aris and Falkner, Matthias},
  title={A Survey on Intent-Based Networking},
  journal={IEEE Communications Surveys \& Tutorials},
  volume={25},
  number={1},
  pages={625--655},
  year={2023},
  doi={10.1109/COMST.2022.3215919}
}

@inproceedings{cegis,
  author={Solar-Lezama, Armando and Tancau, Liviu and Bodik, Rastislav and Seshia, Sanjit and Saraswat, Vijay},
  title={Combinatorial Sketching for Finite Programs},
  booktitle={Proc. 12th Int. Conf. Architectural Support for Programming Languages and Operating Systems},
  pages={404--415},
  year={2006},
  doi={10.1145/1168857.1168907}
}

@article{ndtsurvey,
  author={Kuruvatti, Nandish P. and Habibi, Mohammad Asif and Partani, Sanket and Han, Bin and Fellan, Amina and Schotten, Hans D.},
  title={Empowering {6G} Communication Systems With Digital Twin Technology: A Comprehensive Survey},
  journal={IEEE Access},
  volume={10},
  pages={112158--112186},
  year={2022},
  doi={10.1109/ACCESS.2022.3215493}
}

@inproceedings{multicommodity,
  author={Fortz, Bernard and Thorup, Mikkel},
  title={Internet Traffic Engineering by Optimizing {OSPF} Weights},
  booktitle={Proc. IEEE INFOCOM},
  pages={519--528},
  year={2000},
  doi={10.1109/INFCOM.2000.832232}
}

@article{kshortest,
  author={Yen, Jin Y.},
  title={Finding the {$K$} Shortest Loopless Paths in a Network},
  journal={Management Science},
  volume={17},
  number={11},
  pages={712--716},
  year={1971}
}

@inproceedings{ibnconflict,
  author={Ojaghi, Behnam and Vilalta, Ricard and Munoz, Ra{\'u}l},
  title={Optimizing Throughput and Energy with Intent-Based {6G} Network Slicing: Conflict Handling},
  booktitle={Proc. IEEE International Conference on Communications},
  year={2025}
}

@article{surveyAD,
  author={Liu, Fei and Yao, Yiming and Guo, Ping and Yang, Zhiyuan
          and Lin, Xi and Zhao, Zhe and Tong, Xialiang and Mao, Kun
          and Lu, Zhichao and Wang, Zhenkun and Yuan, Mingxuan
          and Zhang, Qingfu},
  title={A Systematic Survey on Large Language Models for Algorithm Design},
  journal={ACM Computing Surveys},
  volume={58},
  number={8},
  pages={1--32},
  articleno={218},
  year={2026},
  doi={10.1145/3787585}
}

@inproceedings{intentpaper,
  author={Ojaghi, Behnam and Vilalta, Ricard and Mu{\~n}oz, Ra{\'u}l},
  title={Intent-Based Network Resource Slicing in {6G}},
  booktitle={Proc. IEEE Network of the Future},
  year={2024},
  doi={10.1109/NoF62948.2024.10741513}
}

@article{ibnstnsm,
  author={Ojaghi, Behnam and Vilalta, Ricard and Mu{\~n}oz, Ra{\'u}l},
  title={{IBNS}: Optimizing Intent-Based {6G} Network Slicing for Conflict
         Detection and Mitigation},
  journal={IEEE Transactions on Network and Service Management},
  volume={23},
  pages={2818--2831},
  year={2026},
  doi={10.1109/TNSM.2026.3668027}
}

@article{m2oibns,
  author={Ojaghi, Behnam and Mu{\~n}oz, Ra{\'u}l and Vilalta, Ricard},
  title={Intent-Based Network Slicing in {6G}: Optimization Approach for
         Conflict Resolution},
  journal={IEEE Transactions on Mobile Computing},
  year={2026},
  doi={10.1109/TMC.2026.3677566}
}

\end{document}